\documentclass{article}

\usepackage{comment}
\usepackage[utf8]{inputenc} % allow utf-8 input
\usepackage[T1]{fontenc}    % use 8-bit T1 fonts
\usepackage{hyperref}       % hyperlinks
\usepackage{url}            % simple URL typesetting
\usepackage{booktabs}       % professional-quality tables
\usepackage{amsfonts}       % blackboard math symbols
\usepackage{nicefrac}       % compact symbols for 1/2, etc.
\usepackage{microtype}      % microtypography
\usepackage{lipsum}
\usepackage{fancyhdr}       % header
\usepackage{graphicx}       % graphics
\graphicspath{{media/}}     % organize your images and other figures under media/ folder

\usepackage{lineno,hyperref}
\usepackage{geometry} % This package allows the editing of the page layout
\usepackage{amsmath}  % This package allows the use of a large range of mathematical formula, commands, and symbols
\usepackage{amsfonts}
\usepackage{forest} % This package allows tree diagrams
\usepackage{graphicx}  % This package allows the importing of images
\usepackage{amssymb}
\usepackage{bbm}
\usepackage{setspace}
\usepackage{soul}

\usepackage{amsthm} % This package allows the use of theorems and definitions

\newtheorem{theorem}{Theorem}
\newtheorem{prop}{Proposition}
\newtheorem{lemma}{Lemma}

\usepackage{enumitem}

\usepackage[round]{natbib}
\title{Stochastic dynamics of two-compartment cell proliferation models with regulatory mechanisms for hematopoiesis}

\author{
  Ren-Yi Wang$^{1}$, Marek Kimmel$^{2}$, and Guodong Pang$^{3}$ \\
  $^{1,2}$\text{Department of Statistics, Rice University, Houston, TX, 77005, USA}\\
  $^{3}$\text{Department of Computational Applied Mathematics and Operations Research,}\\ \text{Rice University, Houston, TX, 77005, USA } \\
  $^{1,2,3}$\texttt{rw47,kimmel,gdpang@rice.edu} \\
}

\begin{document}
\maketitle
\allowdisplaybreaks

\begin{abstract}
    We present an asymptotic analysis of a stochastic two-compartmental cell division system with regulatory mechanisms inspired by \cite{getto2013global}. The hematopoietic system is modeled as a two-compartment system, where the first compartment consists of dividing cells in the bone marrow, referred to as type $0$ cells, and the second compartment consists of post-mitotic cells in the blood, referred to as type $1$ cells. Division and self-renewal of type $0$ cells are regulated by the population density of type $1$ cells. By scaling up the initial population, we demonstrate that the scaled dynamics converges in distribution to the solution of a system of ordinary differential equations (ODEs). This system of ODEs exhibits a unique non-trivial equilibrium that is globally stable. Furthermore, we establish that the scaled fluctuations of the density dynamics converge in law to a linear diffusion process with time-dependent coefficients. When the initial data is Gaussian, the limit process is a Gauss-Markov process. We analyze its asymptotic properties to elucidate the joint structure of both compartments over large times. This is achieved by proving exponential convergence in the 2-Wasserstein metric for the associated Gaussian measures on an $L^{2}$ Hilbert space. Finally, we apply our results to compare the effects of regulating division and self-renewal of type $0$ cells, providing insights into their respective roles in maintaining hematopoietic system stability.
\end{abstract}

% keywords can be removed

\section{Introduction}
Hematopoiesis is the process of blood cell production. A hematopoietic stem cell (HSC) can proliferate, differentiate, or die. When an HSC differentiates, it becomes a multipotent progenitor cell, which further differentiates to produce mature cells. In \cite{getto2013global}, the authors refer to all dividing cells in the bone marrow as stem cells and all differentiated cells in the blood as mature cells. To avoid confusion between HSCs and type $0$ cells, we refer to all dividing cells as type $0$ cells and all differentiated cells as type $1$ cells. Regulatory mechanisms are necessary to achieve homeostasis (the stability of the cell production system). Specifically, we model the regulations based on extracellular signaling molecules produced by type $1$ cells. For instance, cytokines can be produced by granulocytes (\cite{getto2013global}) and megakaryocytes (\cite{bruns2014megakaryocytes}) to regulate type $0$ cell proliferation. For a more detailed biological exposition of regulatory mechanisms in the hematopoietic system, we refer readers to \cite{hurwitz2020hematopoietic}.

Numerous mathematical models have been developed to investigate various aspects of hematopoiesis. For example, \cite{getto2013global} studied the global stability of two-compartment models with specific regulatory mechanisms, while \cite{marciniak2009modeling} conducted simulations to assess the efficiency of different regulatory strategies in multi-compartment models. The qualitative behavior of hematopoiesis models, formulated as systems of nonlinear delay differential equations, was analyzed in \cite{arino1986stability}. The dynamics of hematopoietic system oscillations have been examined in \cite{knauer2020oscillations} and \cite{bonnet2021combined}. Furthermore, \cite{bonnet2021large} introduced a stochastic model, devoid of regulatory mechanisms, to explain unexpected fluctuations in mature blood cell counts.

Asymptotic analysis of deterministic models for hematopoiesis with regulatory mechanisms have been extensively studied as mentioned above,  while stochastic models with regulatory mechanisms have received comparatively less attention (see, e.g., \cite{lomeli2021optimal} and references therein).
In this paper, we aim to bridge the gap between deterministic and stochastic modeling of hematopoiesis with regulatory mechanisms. We propose a stochastic model and generalize the regulatory mechanisms described in \cite{getto2013global}. Our stochastic framework offers three key advantages. First, when two deterministic models have similar fits to data, it is of advantage to return to the underlying stochastic models to investigate if they differ with respect to the level of fluctuations in transient and equilibrium states. By comparing these fluctuations with observed data, we can distinguish competing models. Second, in the deterministic framework, regulation is deemed efficient if the cell dynamics converge rapidly to the steady state. In the stochastic framework, we introduce an additional measure of effectiveness: a regulation is effective if fluctuations around the steady state remain small. Therefore, it is possible for a regulatory mechanism to be efficient but not effective, and vice-versa. Finally, by analyzing the self-dependence (autocorrelation function) of the dynamics, we can investigate the joint structure and range of dependence in the process.

We view hematopoiesis as a two-compartment model. The first compartment contains type $0$ cells and the second compartment contains type $1$ cells. For $i \in \{0,1\}$, let $N_{i}^{(r)}(t)$ denote the number of type $i$ cells at time $t$ with $N_{i}^{(r)}(0) = n_{i}^{(r)}>0$. The parameter $r$ is a scaling parameter such that 
\begin{align*}
    \bar{n}_{i}^{(r)} = \frac{n_{i}^{(r)}}{r} \to \bar{n}_{i} > 0;\quad i \in \{0,1\}.
\end{align*}
If we view scaling up $r$ as scaling up the physical space where type $0$ and type $1$ cells reside, $\bar{n}_{i}^{(r)}$ can be interpreted as density of type $i$ cells at time $0$. Denote the rate of division and the probability of self-renewal of type $0$ cells by $p$ and $a$, respectively. Then, a type $0$ cell proliferates at rate $p \cdot a$ and differentiates at rate $p \cdot (1-a)$. We model the cell cycle length using exponential distributions. The validity of this assumption in the context of hematopoiesis is discussed in detail in Section $\ref{discussion}$. We study two types of regulatory mechanisms in this paper. The first model regulates the division rate $p$ by letting it depend on the density of type $1$ cells ($\Bar{N}_{1}^{(r)}=N_{1}^{(r)}/r$). The second model regulates the self-renewal probability $a$, which also depends on $\Bar{N}_{1}^{(r)}$. Similar to many biological processes, our model can be formulated as a stochastic chemical reaction network; we refer readers to \cite{anderson2015stochastic} for the general theory.

We derive large-scale approximations for density-dependent dynamics using the standard theory presented in \cite{ethier2009markov}. For a more comprehensive discussion on functional limit theorems for density-dependent dynamics, we refer readers to \cite{kurtz1981approximation} and a more recent work by \cite{prodhomme2023strong}. Under conditions on $p$ and $a$ specified in Section \ref{model and assumption}, we derive deterministic approximations (FLLN) for the density dynamics of $\bar{\mathbf{N}}^{(r)} = \mathbf{N}^{(r)}/r$ for large $r$ in Section \ref{FLLN}. The FLLN limit indicates that the density dynamics $\Bar{\mathbf{N}}^{(r)}$ has a deterministic central tendency $\Bar{\mathbf{N}}$ as $r \to \infty$. The central tendency $\Bar{\mathbf{N}}$ satisfies a system of ordinary differential equations (ODEs) of the form $\bar{\mathbf{N}}' = \mathbf{F}(\bar{\mathbf{N}}); \bar{\mathbf{N}}(0) = \bar{\mathbf{n}}$. This system has a non-trivial equilibrium $\bar{\mathbf{N}}^{*}$ that is globally stable. More specifically, we show for any fixed initial condition $\bar{\mathbf{n}} > 0$, $\|\bar{\mathbf{N}}(t)-\bar{\mathbf{N}}^{*}\|_{2} = o(\exp(-|\eta| t))$ as $t\to\infty$ for all $\eta \in (\lambda,0)$, where $\lambda$ is the maximal real part of $\nabla\mathbf{F}(\bar{\mathbf{N}}^{*})$'s eigenvalues. In \cite{getto2013global}, global stability was demonstrated using Lyapunov functions. In this paper, we instead apply the Poincaré–Bendixson theorem and rule out undesired cases, as constructing Lyapunov functions is more challenging for the broader class of models considered here. 

In Section $\ref{FCLT}$, we consider the scaled fluctuation around the central tendency, denoted by $\hat{\mathbf{N}}^{(r)} = \sqrt{r}(\bar{\mathbf{N}}^{(r)}-\bar{\mathbf{N}})$. We show the sequence of processes $\hat{\mathbf{N}}^{(r)}$ converges weakly to $\hat{\mathbf{N}}$ (the FCLT), which is the solution to a linear diffusion with time-dependent coefficients. We show that the time-dependent coefficients of the limiting diffusion $\hat{\mathbf{N}}$ become constant when the FLLN limit starts at its equilibrium $\bar{\mathbf{N}}^{*}$ in Section \ref{fclt eq}, resulting in a stationary linear diffusion process (Ornstein-Uhlenbeck process) c.f. \cite{karatzas2012brownian}. Consequently, the covariance function of the limiting diffusion under the FLLN equilibrium converges to a constant $V^{*}$, which can be solved by a linear equation after specifying a regulatory mechanism (the functional form of $p$ and $a$). Next, we show that the limiting diffusion $\hat{\mathbf{N}}$ at large times resembles a linear diffusion $\mathbf{G}^*$ with mean function $\mathbf{m}(t) = \mathbf{0}$. For $s \leq t$, the autocovariance function depends on $s$ and $t$ only through $t-s$. Specifically, we show that the time-shifted FCLT dynamics $\mathbf{G}_T(\cdot) = \hat{\mathbf{N}}(T+\cdot)$ converges to $\mathbf{G}^*(\cdot)$ over $[0,M]$ for arbitrary $M>0$ as $T\to\infty$, where we cast these processes to a space of Gaussian measures on a $L^2$ Hilbert space and use the $2-$Wasserstein metric. We derive the explicit $2-$Wasserstein metric for the Gaussian measures corresponding to $\mathbf{G}_T(t)$ and $\mathbf{G}^*$ and show that convergence is exponentially fast with rate $|\eta|/2$, where $\eta \in (\lambda,0)$. Convergence of the $2-$Wasserstein distance offers new insights into the joint structure and self-dependence of the stochastic two-compartment models, aspects that have not been explored in previous studies.
 
In Section \ref{appliactions}, we apply our results to two specific regulatory mechanisms proposed in \cite{getto2013global}. The authors of \cite{getto2013global} showed that regulating self-renewal is more efficient than regulating division within the ODE framework. By analyzing the stochastic dynamics, we reveal that regulating division is more effective for type $1$ populations, whereas regulating self-renewal is more effective for type $0$ populations. Finally, we investigate the autocorrelation of the fluctuation dynamics and find that the dynamics with regulation on self-renewal exhibits a shorter range of dependence at large times.

\section{Model and Assumptions} \label{model and assumption}
In this section, we establish stochastic dynamics that generalizes deterministic models in \cite{getto2013global}. Let $N_{0}^{(r)}(t)$ be the number of type $0$ cells at time $t$ and $N_{1}^{(r)}(t)$ be the number of type $1$ cells at time $t$ such that for $i \in \{0,1\}$, the initial population sizes $\mathbf{n}^{(r)}$ and density dynamics  $\bar{\mathbf{N}}^{(r)}$ are as follows:
\begin{align*}
    \bar{N}_{i}^{(r)}(t) = \frac{N_{i}^{(r)}(t)}{r}; \quad \bar{N}_{i}^{(r)}(0) = \frac{n_{i}^{(r)}}{r} = \bar{n}_{i}^{(r)} \to \bar{n}_{i} > 0 \text{ as } r\to\infty.
\end{align*}

Denote the division rate by $p(\cdot) > 0$ and self-renewal probability by $a(\cdot) \in (0,1)$. Both parameters are regulated by density of the second compartment. We generalize specific regulatory mechanisms for $p(\cdot)$ and $a(\cdot)$ in \cite{getto2013global} to a class of regulatory mechanisms and identify conditions for global stability of the FLLN dynamics (central tendency of the dynamics).

A type $0$ cell can proliferate (denote this action by $(0,b)$) with rate $p\cdot a$ to increase the type $0$ cell population by $1$, differentiate (denoted by $(0,di)$) with rate $p\cdot (1-a)$ to initiate a clone, or die with rate $\mu$ (denoted $(0,d)$). A type $1$ cell dies at rate $\nu$ (denoted $(1,d)$). Let $\mathcal{A} = \{(0,b),(0,di),(0,d),(1,d)\}$ be the set of actions in the model. 

In the first model, only division is regulated and the dynamics is defined by the following transitions:
\begin{equation} \label{first setup}
\begin{aligned}
    &(N_{0}^{(r)},N_{1}^{(r)}) \to (N_{0}^{(r)}+1,N_{1}^{(r)}) \quad &&\text{ at rate } p(\Bar{N}_{1}^{(r)})\cdot a N_{0}^{(r)} \\
    &(N_{0}^{(r)},N_{1}^{(r)}) \to (N_{0}^{(r)}-1,N_{1}^{(r)}+2) \quad &&\text{ at rate } p(\Bar{N}_{1}^{(r)})\cdot (1-a)N_{0}^{(r)} \\
    &(N_{0}^{(r)},N_{1}^{(r)}) \to (N_{0}^{(r)}-1,N_{1}^{(r)}) \quad &&\text{ at rate } \mu N_{0}^{(r)} \\
    &(N_{0}^{(r)},N_{1}^{(r)}) \to (N_{0}^{(r)},N_{1}^{(r)}-1)  \quad &&\text{ at rate } \nu N_{1}^{(r)}.
\end{aligned}
\end{equation}

We assume $p \in \mathcal{C}^{2}$ and $p'(y)< 0$ for all $y\geq 0$. That is, the division rate decreases as the size of second compartment increases. Additionally, we assume
\begin{align*}
    (2a-1)p(0)>\mu; \quad\lim_{y\to\infty}p(y) =0.
\end{align*}
These conditions are required for uniqueness of FLLN limit's steady state and global stability.

For the second model, only self-renewal probability is regulated, and the stochastic dynamics is defined by
\begin{equation} \label{second setup}
\begin{aligned}
    &(N_{0}^{(r)},N_{1}^{(r)}) \to (N_{0}^{(r)}+1,N_{1}^{(r)}) \quad &&\text{ at rate } p\cdot a(\Bar{N}_{1}^{(r)})N_{0}^{(r)} \\
    &(N_{0}^{(r)},N_{1}^{(r)}) \to (N_{0}^{(r)}-1,N_{1}^{(r)}+2) \quad &&\text{ at rate } p\cdot (1-a(\bar{N}_{1}^{(r)}))N_{0}^{(r)} \\
    &(N_{0}^{(r)},N_{1}^{(r)}) \to (N_{0}^{(r)}-1,N_{1}^{(r)}) \quad &&\text{ at rate } \mu N_{0}^{(r)} \\
    &(N_{0}^{(r)},N_{1}^{(r)}) \to (N_{0}^{(r)},N_{1}^{(r)}-1)  \quad &&\text{ at rate } \nu N_{1}^{(r)}.
\end{aligned}
\end{equation}

Similarly, we assume $a \in \mathcal{C}^{2}$ and $a'(y)< 0$ for all $y\geq 0$. Hence, the probability of self-renewal decreases as the number of type $1$ cells increases. Similar to the first model, we assume
\begin{align*}
    (2a(0)-1)p>\mu; \quad\lim_{y\to\infty} a(y) =0.
\end{align*}

For the second model, an additional assumption is required for stability of the FLLN dynamics: Assume for all $y> 0$, 
\begin{align*}
    \frac{d}{dy}\ln(1-a(y)) < \frac{1}{y}.
\end{align*}

\section{Large-scale Approximations}
We now scale up the initial population sizes to obtain the central tendency (FLLN) of density dynamics and fluctuation around the central tendency (FCLT). More specifically, we show global stability for the FLLN dynamics and study asymptotic behaviors for the FCLT dynamics.

\subsection{Functional Law of Large Numbers} \label{FLLN}
Assume $\frac{\mathbf{n}^{(r)}}{r}\to \Bar{\mathbf{n}} > 0$ as $r\to\infty$ and denote the density dynamics as $\Bar{\mathbf{N}}^{(r)}$ where 
\begin{align*}
    \Bar{N}_{0}^{(r)}=\frac{N^{(r)}_{0}}{r}; \quad \Bar{N}_{1}^{(r)}=\frac{N^{(r)}_{1}}{r}. 
\end{align*}

In Theorem $\ref{thm: FLLN}$ (FLLN), we show $\bar{\mathbf{N}}^{(r)}\to \bar{\mathbf{N}}$ almost surely in space $(\mathbb{D}[0,\infty),d_{\infty}^{\circ})^{2}$, where $d_{\infty}^{\circ}$ is the Skorokhod metric on $\mathbb{D}[0,\infty)$. For properties of this space, we refer readers to Section $16$ of \cite{billingsley2013convergence}. The FLLN limit $\bar{\mathbf{N}}$ is a continuous function satisfying the following initial value problem (IVP),
\begin{align*}
    \bar{\mathbf{N}}(t)' = \mathbf{F}(\bar{\mathbf{N}}(t)); \quad \bar{\mathbf{N}}(0) = \bar{\mathbf{n}}.
\end{align*}

For model $(\ref{first setup})$ and model $(\ref{second setup})$, 
\begin{align*}
    \mathbf{F}(x,y) = \begin{pmatrix}
        (2a-1)p(y)x-\mu x \\
        2(1-a)p(y)x-\nu y
    \end{pmatrix} \quad \text{ and } \quad 
    \mathbf{F}(x,y) = \begin{pmatrix}
        (2a(y)-1)p x-\mu x \\
        2(1-a(y))p x-\nu y
    \end{pmatrix}, \text{ respectively.}
\end{align*}

Denote $A^{*} = \nabla\mathbf{F}(\bar{\mathbf{N}^{*}})$ and define $\lambda=\max\{\Re(\sigma(A^{*}))\}$. We show $\lambda< 0$ in Proposition $\ref{prop: local stability}$, implying the FLLN dynamics is locally asymptotically stable. In Theorem $\ref{thm: global stability}$, we show the FLLN dynamics admits a unique non-trivial steady state and prove global stability. We first show the FLLN dynamics cannot have chaos due to Poincaré–Bendixson theorem. Then we rule out the possibilities of any limit cycle, homoclinic cycle, and heteroclinic cycle to conclude global stability. As a consequence, for any fixed $\eta \in (\lambda,0)$ and initial data $\Bar{\mathbf{n}}$,
\begin{equation} \label{eq: global stability}
    \|\Bar{\mathbf{N}}(t)-\Bar{\mathbf{N}}^{*}\|_{2} = o(\exp(-|\eta|t)) \text{ as } t\to\infty.
\end{equation}

\subsection{Functional Central Limit Theorem} \label{FCLT}
Now we center the dynamics with respect to its FLLN dynamics and scale up the difference. Define the scaled difference $\hat{\mathbf{N}}^{(r)} = \sqrt{r}(\Bar{\mathbf{N}}^{(r)}-\Bar{\mathbf{N}})$ and assume that as $r\to\infty$,
\begin{align*}
    \hat{\mathbf{n}}^{(r)}=\sqrt{r}(\Bar{\mathbf{n}}^{(r)}-\Bar{\mathbf{n}})\Rightarrow \hat{\mathbf{n}} \sim N(\mathbf{u},U).
\end{align*}

Using the standard theory in \cite{ethier2009markov}, we show in Theorem $\ref{thm: FCLT}$ that $\hat{\mathbf{N}}^{(r)}\Rightarrow \hat{\mathbf{N}}$. The FCLT limit process $\Hat{\mathbf{N}}$ is a linear diffusion with time-dependent parameters
\[d\hat{\mathbf{N}}(t) = A(t)\hat{\mathbf{N}}(t)dt + \sigma(t)d\mathbf{B}(t); \quad \hat{\mathbf{N}}(0) = \hat{\mathbf{n}}.\] 
For model $(\ref{first setup})$,
\begin{equation*}
\begin{aligned}
    &A(t) = \nabla \mathbf{F}(\Bar{\mathbf{N}}(t)) = \begin{pmatrix}
        (2a-1)p(\bar{N}_{1}(t))-\mu &(2a-1)p'(\bar{N}_{1}(t))\bar{N}_{0}(t) \\
        2(1-a)p(\bar{N}_{1}(t)) &2(1-a)p'(\bar{N}_{1}(t))\bar{N}_{0}(t) - \nu
    \end{pmatrix};\\
    &\sigma(t) =
    \begin{pmatrix}
        \sqrt{p(\Bar{N}_{1}(t))a\Bar{N}_{0}(t)} &-\sqrt{p(\Bar{N}_{1}(t))(1-a)\Bar{N}_{0}(t)} &-\sqrt{\mu \Bar{N}_{0}(t)} &0 \\
        0 &2\sqrt{p(\Bar{N}_{1}(t))(1-a)\Bar{N}_{0}(t)} &0 &-\sqrt{\nu \Bar{N}_{1}(t)}
    \end{pmatrix}.
\end{aligned}
\end{equation*}
For model $2$, we have 
\begin{equation*}
\begin{aligned}
    &A(t) = \nabla \mathbf{F}(\Bar{\mathbf{N}}(t)) = \begin{pmatrix}
        (2a(\bar{N}_{1}(t))-1)p-\mu &2a'(\bar{N}_{1}(t))p\bar{N}_{0}(t) \\
        2(1-a(\bar{N}_{1}(t)))p &-2a'(\bar{N}_{1}(t))p\bar{N}_{0}(t) - \nu
    \end{pmatrix};\\
    &\sigma(t) =
    \begin{pmatrix}
        \sqrt{pa(\Bar{N}_{1}(t))\Bar{N}_{0}(t)} &-\sqrt{p(1-a(\Bar{N}_{1}(t)))\Bar{N}_{0}(t)} &-\sqrt{\mu \Bar{N}_{0}(t)} &0 \\
        0 &2\sqrt{p(1-a(\Bar{N}_{1}(t)))\Bar{N}_{0}(t)} &0 &-\sqrt{\nu \Bar{N}_{1}(t)}
    \end{pmatrix}.
\end{aligned}
\end{equation*}

Since the initial condition $\Hat{\mathbf{n}}$ is Gaussian, $\Hat{\mathbf{N}}$ is a Gauss-Markov process and it can be characterized by its mean function and autocovariance function as in Section $5.6$ of \cite{karatzas2012brownian}. The mean function $\mathbf{m}$ is defined by
\begin{align*}
    \mathbf{m}(t) = \Phi(t)\mathbf{u}; \quad \Phi'(t) = A(t)\Phi(t), \quad \Phi(0) = I.
\end{align*}
Recall that the initial variance is $U = V(0)$. The autocovariance function $\rho$ has the following form:
\begin{align*}
    \rho(s,t) = \Phi(s)[U + \int_{0}^{s\wedge t}\Phi^{-1}(u)\sigma(u)\sigma^{\top}(u)[\Phi^{-1}]^{\top}(u)du]\Phi^{\top}(t); \quad s,t\geq 0.
\end{align*}

\subsection{Convergence of Mean and Variance} \label{fclt eq} 
If $\Bar{\mathbf{n}} = \Bar{\mathbf{N}}^{*}$, the time-dependent coefficients of the limiting linear diffusion become constants.
\begin{align*}
    A^{*} = \begin{cases}
        \begin{pmatrix}
        0 &(2a-1)p'(\bar{N}_{1}^{*})\bar{N}_{0}^{*} \\
        2(1-a)p(\bar{N}_{1}^{*}) &2(1-a)p'(\bar{N}_{1}^{*})\bar{N}_{0}^{*}-\nu
    \end{pmatrix},  & \text{for model } (\ref{first setup}) \\
    \begin{pmatrix}
        0 &2a'(\bar{N}_{1}^{*})p\bar{N}_{0}^{*} \\
        2(1-a(\bar{N}_{1}^{*}))p &-2a'(\bar{N}_{1}^{*})p\bar{N}_{0}^{*} - \nu
    \end{pmatrix}, &\text{for model } (\ref{second setup}).
    \end{cases}
\end{align*}
The diffusion coefficient has the same form for both models:
\begin{align*}
    \sigma^{*} = \begin{pmatrix}
        \sqrt{\mu\bar{N}_{0}^{*}+\frac{\nu\bar{N}_{1}^{*}}{2}} &-\sqrt{\frac{\nu\bar{N}_{1}^{*}}{2}} &-\sqrt{\mu \Bar{N}_{0}^{*}} &0 \\
        0 &2\sqrt{\frac{\nu\bar{N}_{1}^{*}}{2}} &0 &-\sqrt{\nu \Bar{N}_{1}^{*}}
    \end{pmatrix}.
\end{align*}
We assume $A^{*}$ has distinct eigenvalues, which is equivalent to
\begin{align*}
    \begin{cases}
        -2(1-a)(2a-1)p(\bar{N}_{1}^{*})p'(\bar{N}_{1}^{*})\bar{N}_{0}^{*} \neq [(1-a)p'(\bar{N}_{1}^{*})\bar{N}_{0}^{*}-\frac{\nu}{2}]^{2} \text{, for model } (\ref{first setup}); \\
        -4a'(\bar{\mathbf{N}}_{1}^{*})(1-a(\bar{N}_{1}^{*}))p^{2}\bar{N}_{0}^{*} \neq [a'(\bar{N}_{1}^{*})p\bar{N}_{0}^{*}+\frac{\nu}{2}]^{2} \text{, for model } (\ref{second setup}).
    \end{cases}
\end{align*}

In the remainder of this section, we show convergence of the mean and variance functions for both models. Recall that $\lambda = \max\{\Re(\sigma(A^{*}))\}<0$ for both models. The fundamental matrix $\Phi(t)$ for the mean function $\mathbf{m}(\cdot)$ satisfies the matrix differential equation
\begin{align*}
    \Phi'(t) = A^{*}\Phi(t); \quad \Phi(0) = I_{2}.
\end{align*}
Hence, the mean function is $\mathbf{m}(t) = \exp(A^{*}t)\mathbf{u} = O(e^{-|\lambda|t})$ as $t\to\infty$, which converges to $0$ exponentially fast for each entry. The variance function $V(t)$ satisfies
\begin{align*}
    V'(t) = A^{*}V(t) + [A^{*}V(t)]^{\top} + \sigma^{*}[\sigma^{*}]^{\top}; \quad V(0)=U.
\end{align*}

Define $\Sigma^{*} = \sigma^{*}[\sigma^{*}]^{\top}$ and
\begin{align*}
    W(t) = \begin{pmatrix}
        V_{1,1}(t) \\V_{1,2}(t) \\V_{2,2}(t)
    \end{pmatrix}; \quad
    B^{*} = \begin{pmatrix}
        2A_{1,1}^{*} &2A_{1,2}^{*} &0 \\
        A_{2,1}^{*} &A_{1,1}^{*}+A_{2,2}^{*} &A_{1,2}^{*} \\
        0 &2A_{2,1}^{*} &2A_{2,2}^{*}
    \end{pmatrix}; \quad
    S^{*} = \begin{pmatrix}
        \Sigma_{1,1}^{*} \\
        \Sigma_{1,2}^{*} \\
        \Sigma_{2,2}^{*}
    \end{pmatrix}.
\end{align*}
We rewrite the matrix differential equation as a vector differential equation,
\begin{align*}
    &W'(t) = B^{*}W(t)+S^{*}.
\end{align*}
Since the spectrum of $B^{*}$ is $\sigma(B^{*}) = 2\sigma(A^{*})\cup \{Tr(A^{*})\}$, $\max\{\Re(\sigma(B^{*}))\}=2\lambda < 0$. Hence,
\begin{align*}
    W(t) = \exp(B^{*}t)W(0) + \int_{0}^{t}\exp(B^{*}(t-s))S^{*}ds \to -(B^{*})^{-1}S^{*} \text{ as } t\to\infty.
\end{align*}

Define $W^{*} = -(B^{*})^{-1}S^{*}$, then 
\begin{align*}
    &(W(t)-W^{*})' = B^{*}(W(t)-W^{*}) \\
    \Longrightarrow \quad &\|W(t)-W^{*}\|_{2}= O(\exp(-2|\lambda|t))\text{ as } t\to\infty.
\end{align*}
From the definition of $W^{*}$, we deduce the limit $V^{*}=\lim_{t\to\infty}V(t)$ for the variance function satisfies 
\begin{equation}\label{eq: limiting variance}
    A^{*}V^{*}+[A^{*}V^{*}]^{\top}=-\Sigma^{*}.
\end{equation} 

Let $\|\cdot\|_{F}$ denote the Frobenius norm for matrices. In Lemma $\ref{lemma: variance convergence}$, we lift the assumption $\bar{\mathbf{n}} = \bar{\mathbf{N}}^{*}$ and show as $t\to\infty$,
\begin{align*}
    &\mathbf{m}(t) = O(\exp(-|\lambda|t));\\
    &\|V(t)-V^{*}\|_{F}= o(\exp(-|\eta|t)), \quad \forall \eta \in (\lambda,0).
\end{align*}
Notice that the rate of convergence becomes smaller (from $|2\lambda|$ to $|\eta|$) for the variance function.

\section{Large-time Approximation} \label{Wasserstein}
In the previous section, we have shown that the variance function for the FCLT dynamics converges and studied the rate of convergence. In this section, we extend the result by showing that the FCLT dynamics at large times resembles a Gaussian process $\mathbf{G}^{*}$ with a mean function $\mathbf{m}^{*}(t) = \mathbf{0}$ and an autocovariance function $K^{*}(s,t)=V^{*}\exp((t-s)[A^{*}]^{\top})$ for $s\leq t$. Let $T \in \mathbb{N}$ and define $\mathbf{G}_{T}(t)=\Hat{\mathbf{N}}(T+t)$ with mean function $\mathbf{m}_{T}$ and autocovariance function $K_{T}$, where
\begin{equation*}
    \mathbf{m}_{T}(t) = \Phi(T+t)\mathbf{u}; \quad K_{T}(s,t) = \rho(T+s,T+t).
\end{equation*}

We show that $\mathbf{G}_{T}$ converges to $\mathbf{G}^{*}$ as $T\to\infty$. To introduce the notion of convergence, we follow Section $5$ in \cite{minh2023entropic} to establish a one-to-one correspondence between the space of Gaussian processes and the space of Gaussian measures on $\mathcal{H}=L^{2}([0,M],\mathcal{B}([0,M]),Leb)$, where $M>0$ is fixed. $\mathcal{H}$ is a separable Hilbert space. For a Gaussian process with mean function $\mathbf{m}\in \mathcal{H}$ and autocovariance function $K(s,t)$ being continuous, symmetric, and positive-definite, we invoke Mercer's theorem to induce a Gaussian measure $\mathcal{N}(\mathbf{m},C_{K})$ on $\mathcal{H}$. $C_{K}:\mathcal{H}\to\mathcal{H}$ is the covariance operator defined by 
\begin{align*}
    (C_{K})f(s)=\int_{[0,M]}K(s,t)f(t)dt.
\end{align*}
The covariance operator $C_{K}$ is of trace class, self-adjoint, and positive.

Since the mean functions $\mathbf{m}_{T}, \mathbf{m}^{*}\in \mathcal{H}$ and the autocovariance functions $K_{T}, K_{*}$ are continuous, symmetric, and positive-definite, we induce a sequence of Gaussian measures $\mathcal{G}_{T}\sim\mathcal{N}(\mathbf{m}_{T}, C_{T})$ and $\mathcal{G}^{*}\sim\mathcal{N}(\mathbf{0},C_{*})$. The square of the $2-$Wasserstein distance is
\begin{align*}
    W_{2}^{2}(\mathcal{G}_{T},\mathcal{G}_{*}) &= \|\mathbf{m}_{T}\|_{\mathcal{H}}^{2}+Tr(C_{T}+C_{*}-2(C_{*}^{1/2}C_{T}C_{*}^{1/2})^{1/2}).
\end{align*}

In Theorem $\ref{theorem4}$, we show for any fixed $M$, $W_{2}(\mathcal{G}_{T},\mathcal{G}_{*})=o(\exp(\frac{\eta}{2}T))$ as $T\to\infty$ for all $\eta \in (\lambda,0)$. This is done by first applying operator monotonicity for the square root function to obtain an upper bound for the trace term. Let $\|\cdot\|_{Tr}$ denote the trace norm and $\|\cdot\|$ denote the Hilbert-Schmidt norm, we show
\begin{align*}
    \|C_{T}+C_{*}-2(C_{*}^{1/2}C_{T}C_{*}^{1/2})^{1/2}\|_{Tr} \leq \|(C_{T}+C_{*})^{2}-4(C_{*}^{1/2}C_{T}C_{*}^{1/2})\|_{Tr}^{\frac{1}{2}}= \|C_{T}-C_{*}\|_{HS}.
\end{align*}
Denote the space of Hilbert-Schmidt operators on $\mathcal{H}$ by $B_{HS}(\mathcal{H})$, then we have the following isometric isomorphism:
\begin{align*}
    B_{HS}(\mathcal{H}) \cong \mathcal{H}^{*}\otimes \mathcal{H} \cong \mathcal{H}\otimes \mathcal{H} \cong L^{2}([0,M]^{2}; (\mathbb{R}^{2},\|\cdot\|_{2})\otimes(\mathbb{R}^{2},\|\cdot\|_{2})),
\end{align*}
where $\mathcal{H}^{*}$ is the dual space of $\mathcal{H}$. The last norm is on the space of symmetric matrix functions, and we use convergence results in Section $\ref{fclt eq}$ to conclude exponential convergence of the $2-$Wasserstein distance.

To interpret this result, we focus on $\Hat{\mathbf{N}}$ in a moving interval $[T,T+M]$. As we increase $T\to\infty$, the dynamics of $\Hat{\mathbf{N}}$ will become more and more similar to that of $\mathbf{G}^{*}$. The length of the interval $M$ may be arbitrarily large.

\section{Applications} \label{appliactions}
We apply our framework to study the stochastic dynamics of two regulatory mechanisms proposed in \cite{getto2013global}. Since all the results rely on $r\to \infty$, we fix a large $r = R = 2000$ to approximate limiting results in this section. Let $\kappa >0$ be the strength of the regulatory mechanism. As in Section $\ref{model and assumption}$, we require $p > 0$ and $a \in (0,1)$. We set up two regulatory mechanisms:
\begin{align*}
    &\textbf{model $(\ref{first setup})$: } p(y) = \frac{p}{1+\kappa y}; \qquad \textbf{model $(\ref{second setup})$: } a(y) = \frac{a}{1+\kappa y}.
\end{align*}
Using parameters in Table $\ref{table:1}$, it is easy to verify both models satisfy conditions in Section $\ref{model and assumption}$.

\renewcommand{\arraystretch}{1.2}
\begin{table}[!h]
\centering
\begin{tabular}{|p{1.2cm}|p{1cm}|p{1cm}|p{0.8cm}|p{0.8cm}|p{1.8cm}|p{1.2cm}|}
    \hline
    \multicolumn{7}{|c|}{Model Parameters} \\
    \hline
    &$p$ &$a$ &$\mu$ &$\nu$ &$\kappa$ &$-\lambda$ \\
    \hline 
    Model $1$ &0.827  &0.589 &0.1 &2.3 &$1.19\cdot 10^{-9}\cdot R$ &0.024  \\ 
    Model $2$ &0.56  &0.87  &0.1 &2.3 &$1.19\cdot 10^{-9}\cdot R$ &0.618  \\
    \hline
\end{tabular}
\caption{$\lambda$ is the real part of the dominating eigenvalue for matrix $A^{*}$, representing the efficiency of the regulatory mechanism.}
\label{table:1}
\end{table}
As we see in Table $\ref{table:1}$, $|\lambda|$ for model $2$ is larger than that of model $1$, indicating that regulating self-renewal is more efficient. This conclusion is also supported by Figures $\ref{model1density}$ and $\ref{model2density}$. Under identical theoretical steady states and initial conditions, model $(\ref{first setup})$ takes approximately 200 days to reach the steady state, while model $(\ref{second setup})$ stabilizes in about 10 days. Hence, the rate of convergence provides insight into the regulated parameter. However, to estimate the rate of converge, one needs to conduct experiments to perturb the dynamics from its steady state and collect time-series data. In the following section, we provide an alternative solution to this problem by investigating the fluctuation of dynamics.

Figure $\ref{variance}$ highlights significant differences in the limiting variances of the two models. The limiting variance for type $0$ cell density in model $1$ is much larger than that of model $2$. In contrast, the limiting variance in model $2$ is slightly larger than that in model $1$ for type $1$ cell density. These observations suggest that estimating variance at the steady state can help identify the regulatory mechanism in action. 

Finally, Figure $\ref{autocorr}$ demonstrates marked differences in the autocorrelation functions between the two models. Autocorrelation functions represent dependence between one compartment and its past or other compartment's past. It is defined by 
\begin{align*}
    Corr(\hat{N}_{i}(s),\hat{N}_{j}(s+t)) = \frac{K_{i,j}^{*}(s,s+t)}{\sqrt{V^{*}_{i,i}V^{*}_{j,j}}} = \frac{K_{i,j}^{*}(0,t)}{\sqrt{V^{*}_{i,i}V^{*}_{j,j}}}.
\end{align*}
All autocorrelation functions decay to zero as $t\to\infty$, implying fluctuations of type $0$ and type $1$ densities will eventually ``forget'' their history. Autocorrelation functions in model $2$ decay much faster than that in model $1$, suggesting that dynamics with self-renewal regulation have a shorter range of dependence. Additionally, correlations in model $2$ can be negative, while all correlations in model $1$ are positive. These distinctions provide another dimension for differentiating regulatory mechanisms.

\begin{figure}[!h]
    \centering
    \includegraphics[width=\textwidth]{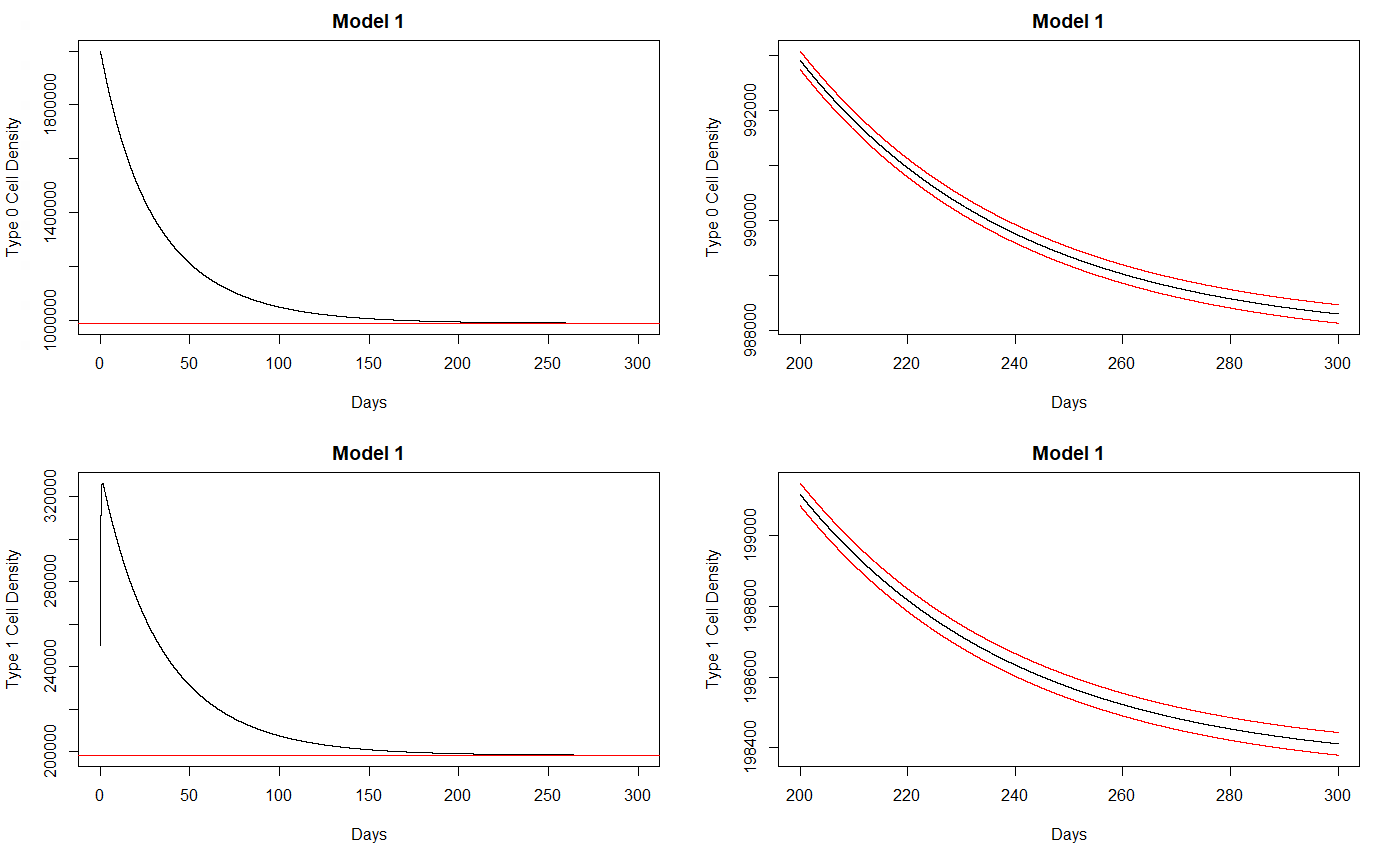}
    \caption{FLLN dynamics for model $1$. Left column: density dynamics (black) and theoretical steady states (red). Right column: density dynamics and $95\%$ confidence band.}
    \label{model1density}
\end{figure}

\begin{figure}[!h]
    \centering
    \includegraphics[width=\textwidth]{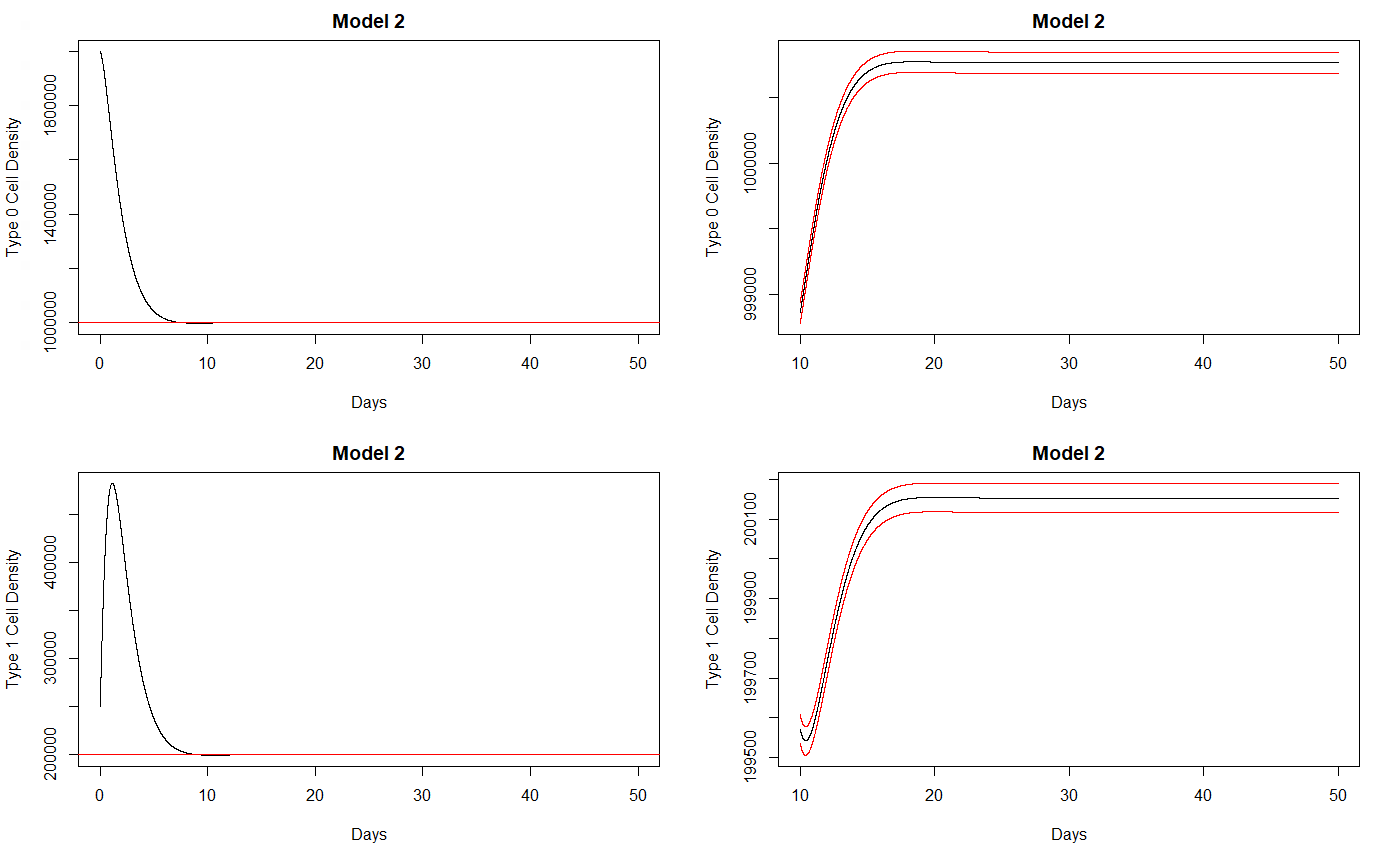}
    \caption{FLLN dynamics for model $2$. Left column: density dynamics (black) and theoretical steady states (red). Right column: density dynamics and $95\%$ confidence band.}
    \label{model2density}
\end{figure}

\begin{figure}[!h]
    \centering
    \includegraphics[width=\textwidth]{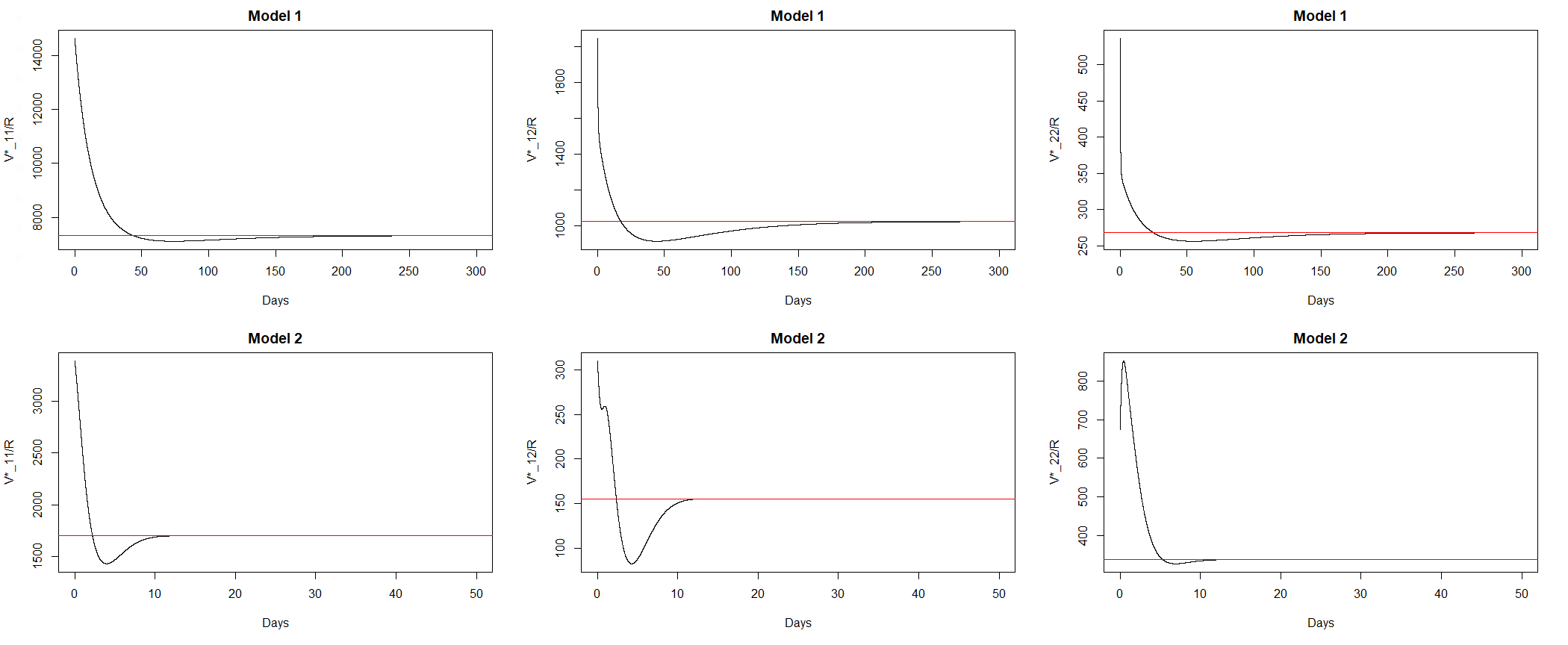}
    \caption{Variance function for both models with initial conditions $\mathbf{u} = \mathbf{0}$ and $U = 2\cdot V^{*}$. Red lines are components in the theoretical limiting variance $V^{*}$.}
    \label{variance}
\end{figure}

\begin{figure}[!h]
    \centering
    \includegraphics[width=\textwidth]{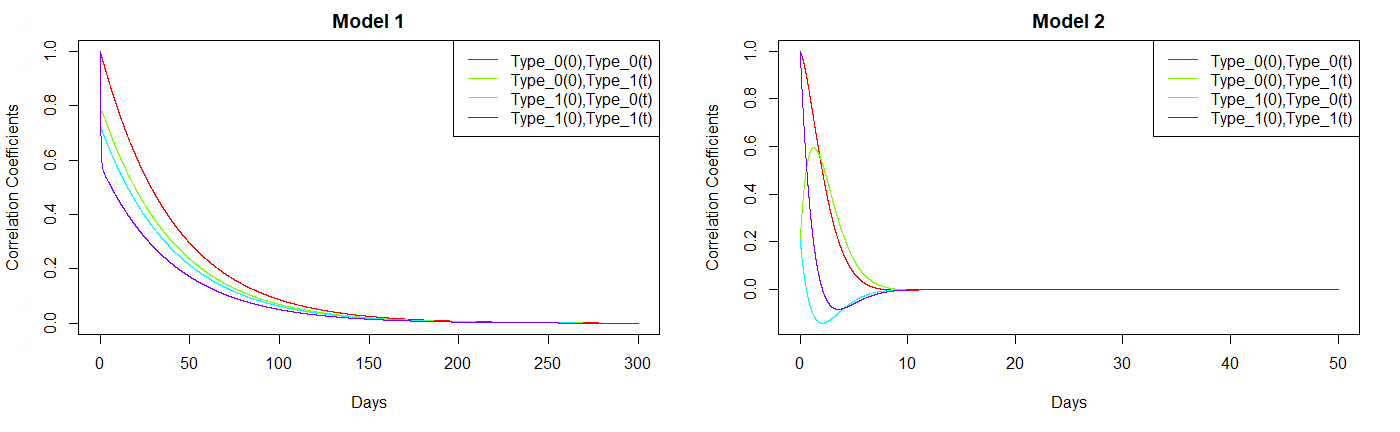}
    \caption{Autocorrelation function for both models. Since for $s \leq t$, the autocovaraince function $K^{*}(s,t)$ depends on $s$ and $t$ through $t-s$, it suffices to plot the autocorrelation function between $\hat{N}_{i}(0)$ and $\hat{N}_{j}(t)$, $i,j \in \{0,1\}$.}
    \label{autocorr}
\end{figure}

\clearpage

\section{Discussions} \label{discussion}
In this paper, we model hematopoiesis with regulatory mechanisms as a state-dependent branching process with exponential lifetimes, making the process Markovian. While many cell proliferation systems exhibit non-exponential lifetimes, the use of exponential distributions may be justified in specific contexts. Laboratory studies show that sojourn times in various cell cycle phases are finite, and some bacteria exhibit nearly synchronous cell divisions, as extensively documented in \cite{kimmel2015branching}. However, the dynamics differ in cancer and physiological control systems. For instance, lung cancer cell populations double on the order of weeks \cite{gorlova2005estimating}, and hematopoietic stem cells may take even longer \cite{catlin2011replication}. Much of this time is spent awaiting ``permission'' to divide (in healthy cells) or undergoing multiple rounds of DNA repair (in cancer cells). This part of the cell cycle is plausibly modeled by exponential distributions, although in vivo measurements are challenging. Notably, models of physiological stem cell systems based on exponential cell lifetimes, which correspond to ordinary differential equations for expected values, have demonstrated significant accuracy. For example, such models effectively describe telomere shortening in hematopoietic stem cells, as supported by experimental and theoretical studies in \cite{sidorov2009leukocyte}.

We analyze stochastic dynamics of a given regulatory mechanism on the two-compartment model. We first scale up initial population sizes to study the large-scale approximation (FLLN and FCLT) of the dynamics. Then, we pass with time to infinity to study the long-time asymptotic behavior of FLLN and FCLT. We study the stability property of the FLLN and FCLT dynamics. More specifically, we show global stability for the non-trivial steady state and study the rate of convergence to the steady state. For the FCLT dynamics, we show exponential convergence in $2-$Wasserstein metric to a stationary Gauss-Markov process. To interpret the results, we view the rate of convergence of FLLN to the non-trivial steady state as the efficiency of the regulatory mechanism. The FCLT limit at large times represents fluctuation around equilibrium and we associate the limiting variance of fluctuation with effectiveness of the regulatory mechanism. For a regulatory mechanism to be effective, the limiting variance should be small. 

For two specific regulatory mechanisms in \cite{getto2013global}, the authors already show regulating self-renewal is more efficient than regulating division. Hence, regulating self-renewal protects the dynamics against perturbation from its steady state. Building on the work of \cite{getto2013global}, we provide a more detailed comparison between these two regulatory strategies. Specifically, we show that regulating division is more effective for type $1$ cells, while regulating self-renewal is more effective for type $0$ cells. Since the primary goal of hematopoiesis is to produce mature cells to sustain biological functions, regulating division ensures a steadier supply of these cells.

We conclude by mentioning possible extensions for our framework. In the present paper, we adopt the assumptions in \cite{getto2013global} and only investigate cases when one parameter is regulated. Standard theory in \cite{ethier2009markov} for deriving FLLN and FCLT dynamics works for more general regulatory mechanisms. For instance, it is possible to consider regulating both division and self-renewal. Naturally, proving global stability of FLLN dynamics becomes harder for more general regulatory mechanisms. Notice that the conditions in Section $\ref{model and assumption}$ are imposed to guarantee stability of the FLLN and FCLT dynamics. Depending on the biological context, the instability of FLLN and FCLT dynamics might be of interest. Finally, extension to multi-compartment models is conceptually straightforward, but complication levels up as the number of compartment increases.

\section{Acknowledgments}
Ren-Yi Wang was supported by the Department of Statistics at Rice University and by the NIH grant R01 268380 to Wenyi Wang. Marek Kimmel was supported in part by the NSF DMS RAPID COLLABORATION grant 2030577 to  Marek Kimmel and Simon Tavar\'e, and by the National Science Center (Poland) grant number 2021/41/B/NZ2/04134 to Marek Kimmel. We thank the editors and reviewers for their precious comments on this paper.

\section{Conflict of interest}
The authors have no competing interests.

\section{Data Availability Statements}
Data sharing not applicable to this article as no datasets were generated or analyzed during the current study.

\clearpage
\appendix
\section{Appendix}
In this appendix, we let $\|\cdot\|_{F}$ denote Frobenius norm for matrices and $\|\cdot\|_{2}$ denote the Euclidean norm for vectors. We use $``\Rightarrow"$ for convergence in distribution (weak convergence) and $``\Longrightarrow"$ for implication. For properties of Skorokhod space $(\mathbb{D}[0,\infty),d_{\infty}^{\circ})$, we refer readers to Section $16$ of \cite{billingsley2013convergence}.
\subsection{Construction of Dynamics}
In this section, we construct both models using the Poisson process construction in \cite{ethier2009markov}. 

Denote the division rate by $p(\cdot) > 0$ and self-renewal probability by $a(\cdot) \in (0,1)$. Both parameters are regulated by density of the second compartment. A type $0$ cell can proliferate (denote this action by $(0,b)$) with rate $p\cdot a$ to increase the type $0$ cell population by $1$, differentiate (denoted by $(0,di)$) with rate $p\cdot (1-a)$ to initiate a clone, or die with rate $\mu$ (denoted $(0,d)$). A type $1$ cell dies at rate $\nu$ (denoted $(1,d)$). Let $\mathcal{A} = \{(0,b),(0,di),(0,d),(1,d)\}$ be the set of actions in the model. Let $\{P_{a}^{(r)}\mid r\in\mathbb{N},a\in\mathcal{A}\}$ be a set of independent rate $1$ Poisson processes and $n_{i}^{(r)}\in \mathbb{N}$ for $r \in \mathbb{N}$ and $i \in \{0,1\}$. The construction of stochastic dynamics corresponding to model $(\ref{first setup})$ is
\begin{align*}
    &N_{0}^{(r)}(t) = n^{(r)}_{0}- P^{(r)}_{(0,d)}\left(\int_{0}^{t}\mu N_{0}^{(r)}(s)ds\right) - P^{(r)}_{(0,di)}\left(\int_{0}^{t}p(\bar{N}_{1}^{(r)}(s))\cdot(1-a)N_{0}^{(r)}(s)ds\right) \\
    &\qquad \qquad +P^{(r)}_{(0,b)}\left(\int_{0}^{t}p(\bar{N}_{1}^{(r)}(s))\cdot a N_{0}^{(r)}(s)ds\right) \\
    &N_{1}^{(r)}(t) = n^{(r)}_{1} - P^{(r)}_{(1,d)}\left(\int_{0}^{t}\nu N_{1}^{(r)}(s)ds\right) + 2\cdot P^{(r)}_{(0,di)}\left(\int_{0}^{t}p(\bar{N}_{1}^{(r)}(s))\cdot(1-a) N_{0}^{(r)}(s)ds\right).
\end{align*}
The construction for model $(\ref{second setup})$ is
\begin{align*}
    &N_{0}^{(r)}(t) = n^{(r)}_{0}- P^{(r)}_{(0,d)}\left(\int_{0}^{t}\mu N_{0}^{(r)}(s)ds\right) - P^{(r)}_{(0,di)}\left(\int_{0}^{t}p\cdot (1-a(\bar{N}_{1}^{(r)}(s)))N_{0}^{(r)}(s)ds\right) \\
    &\qquad \qquad +P^{(r)}_{(0,b)}\left(\int_{0}^{t}p\cdot a(\bar{N}_{1}^{(r)}(s))N_{0}^{(r)}(s)ds\right) \\
    &N_{1}^{(r)}(t) = n^{(r)}_{1} - P^{(r)}_{(1,d)}\left(\int_{0}^{t}\nu N_{1}^{(r)}(s)ds\right) + 2\cdot P^{(r)}_{(0,di)}\left(\int_{0}^{t}p\cdot (1-a(\bar{N}_{1}^{(r)}(s))N_{0}^{(r)}(s)ds\right).
\end{align*}

\subsection{FLLN and FCLT}
For results in this section, we assume 
\begin{align*}
    &\lim_{r\to\infty}\frac{\mathbf{n}^{(r)}}{r} = \bar{\mathbf{n}} > 0 \quad \text{ (entrywise)}; \\
    &\sqrt{r}(\frac{\mathbf{n}^{(r)}}{r}-\bar{\mathbf{n}}) \Rightarrow \hat{\mathbf{n}} \sim N(\mu, U) \quad \text{ as } r\to\infty.
\end{align*}
The FLLN-scaled dynamics $\bar{\mathbf{N}}^{(r)}$ and FCLT-scaled dynamics $\hat{\mathbf{N}}^{(r)}$ are defined by
\begin{align*}
\bar{\mathbf{N}}^{(r)} = \frac{{\mathbf{N}}^{(r)}}{r}\quad  \text{ and } \quad \hat{\mathbf{N}}^{(r)} = \sqrt{r}(\bar{\mathbf{N}}^{(r)}-\bar{\mathbf{N}}), \,\text{ respectively}.
\end{align*}
In addition, we define
\begin{align*}
    &\mathbf{F}(x,y) = \begin{pmatrix}
        (2a-1)p(y)x-\mu x \\
        2(1-a)p(y)x-\nu y
    \end{pmatrix} \quad \text{ for model $(\ref{first setup})$;} \\
    &\mathbf{F}(x,y) = \begin{pmatrix}
        (2a(y)-1)p x-\mu x \\
        2(1-a(y))p x-\nu y   \end{pmatrix}\quad \text{ for model $(\ref{second setup})$.}
\end{align*}
In this section, we derive limits of $\bar{\mathbf{N}}^{(r)}$ and $\hat{\mathbf{N}}^{(r)}$ as $r\to\infty$ using standard tools in \cite{ethier2009markov}.

\begin{theorem} \label{thm: FLLN}
    For both models, $\lim_{r\to\infty}\Bar{\mathbf{N}}^{(r)} \to \Bar{\mathbf{N}}$ almost surely in the Skorohod space $(\mathbb{D}[0,\infty),d_{\infty}^{\circ})^{2}$, where $\Bar{\mathbf{N}}' = \mathbf{F}(\Bar{\mathbf{N}})$ and $\Bar{\mathbf{N}}(0) = \Bar{\mathbf{n}}$.
\end{theorem}
\begin{proof}
    We verify the conditions in Theorem $11.2.1$ in \cite{ethier2009markov}. For model $(\ref{first setup})$,
    \begin{align*}
        \mathbf{F}(x,y) = 
        \begin{pmatrix}
            1 \\ 0
        \end{pmatrix}p(y)ax  +
        \begin{pmatrix}
            -1 \\ 2
        \end{pmatrix}p(y)(1-a)x +
        \begin{pmatrix}
            -1 \\ 0
        \end{pmatrix}\mu x + 
        \begin{pmatrix}
            0 \\ -1
        \end{pmatrix}\nu y.
    \end{align*}

    Since $p$ is a continuous function, for any fixed compact set $K\in \mathbb{R}^{2}$, we have
    \begin{align*}
    \|\begin{pmatrix}
            1 \\ 0
        \end{pmatrix}\|\sup_{(x,y)\in K}\{p(y)ax\}  +
        \|\begin{pmatrix}
            -1 \\ 2
        \end{pmatrix}\|\sup_{(x,y)\in K}\{p(y)(1-a)x\} \\
        +\|\begin{pmatrix}
            -1 \\ 0
        \end{pmatrix}\|\sup_{(x,y)\in K}\{\mu x\} + 
        \|\begin{pmatrix}
            0 \\ -1
        \end{pmatrix}\|\sup_{(x,y)\in K}\{\nu y\}< \infty.
    \end{align*}

    Hence, condition $(2.6)$ of Theorem $11.2.1$ in \cite{ethier2009markov} is satisfied. In addition, since $\mathbf{F}$ is $\mathcal{C}^{2}$, it is locally Lipschitz, so condition $(2.7)$ of Theorem $11.2.1$ in \cite{ethier2009markov} holds. This concludes $\bar{\mathbf{N}}^{(r)}\Rightarrow \bar{\mathbf{N}}$ almost surely in $(\mathbb{D}[0,\infty),d_{\infty})^{2}$. For model $(\ref{second setup})$, we argue analogously.
\end{proof}

\begin{theorem}\label{thm: FCLT}
    For both models, $\Hat{\mathbf{N}}^{(r)}\Rightarrow \Hat{\mathbf{N}}$, where \[d\hat{\mathbf{N}}(t) = \nabla\mathbf{F}(\bar{\mathbf{N}}(t))\hat{\mathbf{N}}(t)dt + \sigma(t)d\mathbf{B}(t)\] and $\hat{\mathbf{N}}(0)\sim N(\mu,U)$. $\mathbf{B}$ is a $4-$dimensional standard Brownian motion. For model $(\ref{first setup})$,
    \begin{align*}
        \sigma(t) =
        \begin{pmatrix}
        \sqrt{p(\Bar{N}_{1}(t))a\Bar{N}_{0}(t)} &-\sqrt{p(\Bar{N}_{1}(t))(1-a)\Bar{N}_{0}(t)} &-\sqrt{\mu \Bar{N}_{0}(t)} &0 \\
        0 &2\sqrt{p(\Bar{N}_{1}(t))(1-a)\Bar{N}_{0}(t)} &0 &-\sqrt{\nu \Bar{N}_{1}(t)}
    \end{pmatrix}.
    \end{align*}

    For model $(\ref{second setup})$,
    \begin{align*}
        \sigma(t) =
        \begin{pmatrix}
        \sqrt{pa(\Bar{N}_{1}(t))\Bar{N}_{0}(t)} &-\sqrt{p(1-a(\Bar{N}_{1}(t)))\Bar{N}_{0}(t)} &-\sqrt{\mu \Bar{N}_{0}(t)} &0 \\
        0 &2\sqrt{p(1-a(\Bar{N}_{1}(t)))\Bar{N}_{0}(t)} &0 &-\sqrt{\nu \Bar{N}_{1}(t)}
    \end{pmatrix}.
    \end{align*}
\end{theorem}

\begin{proof}
    We verify the conditions in Theorem $11.2.3$ of \cite{ethier2009markov}. Although Theorem $11.2.3$ of \cite{ethier2009markov} requires $\hat{\mathbf{N}}(0)$ to be a constant, it holds for Gaussian initial condition as well. Since $p$ and $a$ are $\mathcal{C}^{2}$, $\nabla\mathbf{F}$ is continuous. For model $(\ref{first setup})$, $p(y)ax, p(y)(1-a)x, \mu x$, and $\nu y$ are continuous. Finally,
    \begin{align*}
    \|\begin{pmatrix}
            1 \\ 0
        \end{pmatrix}\|^{2}\sup_{(x,y)\in K}\{p(y)ax\}  +
        \|\begin{pmatrix}
            -1 \\ 2
        \end{pmatrix}\|^{2}\sup_{(x,y)\in K}\{p(y)(1-a)x\} \\
        +\|\begin{pmatrix}
            -1 \\ 0
        \end{pmatrix}\|^{2}\sup_{(x,y)\in K}\{\mu x\} + 
        \|\begin{pmatrix}
            0 \\ -1
        \end{pmatrix}\|^{2}\sup_{(x,y)\in K}\{\nu y\}< \infty.
    \end{align*}

    Hence, $\hat{\mathbf{N}}^{(r)}\Rightarrow \hat{\mathbf{N}}$ for model $(\ref{first setup})$. For model $(\ref{second setup})$, we conclude by the same argument. 
\end{proof}

%------------------------------------------------------------------------------
\subsection{Stability of the FLLN Dynamics}
To guarantee global stability for the FLLN dynamics, we impose the following regularity conditions. For model $(\ref{first setup})$, we assume $p \in \mathcal{C}^{2}$ and $p'(y)< 0$ for all $y\geq 0$. In addition,
\begin{align*}
    (2a-1)p(0)>\mu; \quad\lim_{y\to\infty}p(y) =0.
\end{align*}
For model $(\ref{second setup})$, we assume $a \in \mathcal{C}^{2}$ and $a'(y)< 0$ for all $y\geq 0$. Furthermore, 
\begin{align*}
    &(2a(0)-1)p>\mu; \quad\lim_{y\to\infty} a(y)=0; \\
    &\frac{d}{dy}\ln(1-a(y)) < \frac{1}{y}, \forall y > 0.
\end{align*}

Since the FLLN dynamics is a two-dimensional autonomous system of ODEs, global stability can be established by the Poincaré-Bendixson theorem. We first show the system is locally stable (Proposition $\ref{prop: local stability}$), then we show the forward orbit through the initial data $\Bar{\mathbf{n}}$ lies in a compact set (Proposition $\ref{prop: bounded trajectory}$). Finally, we conclude global stability by ruling out the undesired cases (Theorem $\ref{thm: global stability}$). 

\begin{prop} \label{prop: local stability}
    $\bar{\mathbf{N}}' = \mathbf{F}(\bar{\mathbf{N}})$ has a unique non-trivial steady state $\bar{\mathbf{N}}^{*}$ for both models. $\bar{\mathbf{N}}^{*}$ is locally asymptotically stable.
\end{prop}
\begin{proof}
    We first show that there exists a unique non-trivial steady state $\bar{\mathbf{N}}^{*}>0$ for both models. For model $(\ref{first setup})$, since $(2a-1)p(0) > \mu$ and $p(y) \downarrow 0$ as $y\to\infty$, there exists a unique $\Bar{N}_{1}^{*}> 0$ such that
    \begin{align*}
        (2a-1)p(\bar{N}_{1}^{*}) - \mu = 0.
    \end{align*}

    This implies 
    \begin{align*}
        \bar{N}_{0}^{*} = \frac{\nu \bar{N}_{1}^{*}}{2(1-a)p(\bar{N}_{1}^{*})} > 0.
    \end{align*}

    We then linearize around $\bar{\mathbf{N}}^{*}$ to obtain
    \begin{align*}
        \nabla\mathbf{F}(\bar{\mathbf{N}}^{*}) = \begin{pmatrix}
            0 &(2a-1)p'(\bar{N}_{1}^{*})\bar{N}_{0}^{*} \\
            2(1-a)p(\bar{N}_{1}^{*}) &2(1-a)p'(\bar{N}_{1}^{*})\bar{N}_{0}^{*}-\nu
        \end{pmatrix}.
    \end{align*}
    Since $p'(\bar{N}_{1}^{*})<0$, we have 
    \begin{align*}
        Tr(\nabla\mathbf{F}(\bar{\mathbf{N}}^{*})) < 0; det(\nabla\mathbf{F}(\bar{\mathbf{N}}^{*}))>0.
    \end{align*}

    For model $(\ref{second setup})$, since $(2a(0)-1)p > \mu$ and $a(y) \downarrow 0$ as $y\to\infty$, there exists a unique $\Bar{N}_{1}^{*}> 0$ such that
    \begin{align*}
        (2a(\bar{N}_{1}^{*})-1)p - \mu = 0.
    \end{align*}

    Hence, 
    \begin{align*}
        \bar{N}_{0}^{*} = \frac{\nu \bar{N}_{1}^{*}}{2(1-a(\bar{N}_{1}^{*}))p} > 0.
    \end{align*}

    Linearize around $\bar{\mathbf{N}}^{*}$, we have
    \begin{align*}
        \nabla\mathbf{F}(\bar{\mathbf{N}}^{*}) = \begin{pmatrix}
            0 &2a'(\bar{N}_{1}^{*}) p\bar{N}_{0}^{*} \\
            2(1-a(\bar{N}_{1}^{*}))p &-2a'(\bar{N}_{1}^{*}) p\bar{N}_{0}^{*}-\nu
        \end{pmatrix}.
    \end{align*}

    By assumption, $\frac{d}{dy}\ln(1-a(y)) < \frac{1}{y}$, which is equivalent to 
    \begin{align*}
        -a'(y) < \frac{1-a(y)}{y}.
    \end{align*}

    Therefore, 
    \begin{align*}
        Tr(\nabla\mathbf{F}(\bar{\mathbf{N}}^{*})) &= -2a'(\bar{N}_{1}^{*}) p\bar{N}_{0}^{*}-\nu \\
        &< 2\frac{1-a(\bar{N}_{1}^{*})}{\bar{N}_{1}^{*}}p\bar{N}_{0}^{*}-\nu = 0.
    \end{align*}
    
    Since $a'(\bar{N}_{1}^{*})<0$, 
    \begin{align*}
        det(\nabla\mathbf{F}(\bar{\mathbf{N}}^{*}))>0.
    \end{align*}

    Therefore, the unique non-trivial steady state $\bar{\mathbf{N}}^{*}$ is locally asymptotically stable for both models.
\end{proof}

\begin{prop}\label{prop: bounded trajectory}
    For any fixed initial condition $\Bar{\mathbf{n}}$, the trajectory of the FLLN limit $\{\Bar{\mathbf{N}}(t)\mid t\geq 0\}$ lies in a compact subset $\mathcal{K}\subset\mathbb{R}^{2}$.
\end{prop}
\begin{proof}
    For model $(\ref{first setup})$, the FLLN dynamics satisfies
    \begin{align*}
        &\bar{N}_{0}' = (2a-1)p(\bar{N}_{1})\bar{N}_{0} - \mu \bar{N}_{0} \\
        &\bar{N}_{1}' = 2(1-a)p(\bar{N}_{1})\bar{N}_{0} - \nu \bar{N}_{1}.
    \end{align*}

    We deduce from the first equation that
    \begin{align*}
        \bar{N}_{0}(t) = \bar{n}_{0}\exp\bigg(\int_{0}^{t}(2a-1)p(\bar{N}_{1}(s))-\mu ds\bigg).
    \end{align*}

    Plug in the second equation, we obtain
    \begin{align*}
        \bar{N}_{1}'(t) &= 2(1-a)p(\bar{N}_{1}(t))\bar{n}_{0}\exp\bigg(\int_{0}^{t}(2a-1)p(\bar{N}_{1}(s))-\mu ds\bigg) - \nu \bar{N}_{1}(t) \\
        &\leq 2(1-a)p(0)\bar{n}_{0}\exp\bigg(\int_{0}^{t}(2a-1)p(\bar{N}_{1}(s))-\mu ds\bigg). 
    \end{align*}

    To show $\bar{N}_{1}(t)$ is bounded, it suffices to show $y(t)$ is bounded, where 
    \begin{align*}
        y'(t) = 2(1-a)p(0)\bar{n}_{0}\exp\bigg(\int_{0}^{t}(2a-1)p(y(s))-\mu ds\bigg).
    \end{align*}

    Suppose the contrary that $y(t)$ is unbounded, then $y(t)\uparrow \infty$ as $t\to\infty$ since $y(t)$ is non-decreasing. Hence, $p(y(t))\downarrow 0$ as $t\to\infty$, so there exists $T> 0$ such that for all $t\geq T$, $(2a-1)p(y(t))-\mu<-\frac{\mu}{2}$. Hence, for all $t\geq T$,
    \begin{align*}
        y'(t) &\leq 2(1-a)p(0)\bar{n}_{0}\exp\bigg(\int_{0}^{T}(2a-1)p(y(s))-\mu ds\bigg)\exp\bigg(-\frac{\mu}{2}(t-T)\bigg) \\
        &= \text{constant}\cdot \exp\bigg(-\frac{\mu}{2}t\bigg).
    \end{align*}

    Therefore, $y(t)$ is bounded and this implies that $\bar{N}_{1}(t)$ is also bounded by some constant $K > 0$. Moreover, since the growth rate of $\bar{N}_{0}$ is lower bounded by $(2a-1)p(K)-\mu$, 
    \begin{align*}
        \bar{N}_{0}(t+\tau) \geq \bar{N}_{0}(t)\exp\bigg([(2a-1)p(K)-\mu]\tau\bigg)\text{ for all } \tau > 0.
    \end{align*}

    Suppose $\bar{N}_{0}(t)$ is unbounded, then for any fixed $\tau> 0$ and $M> 0$, we can find a $t^{*}$ such that for all $s\in [0,\tau]$,
    \begin{align*}
        \bar{N}_{0}(t^{*}+s) \geq M.
    \end{align*}
    This implies for $t \in (t^{*},t^{*}+\tau)$,
    \begin{align*}
        \bar{N}_{1}'(t) \geq 2(1-a)p(K)M - \nu \bar{N}_{1}(t).
    \end{align*}

    By selecting $M$ and $\tau$ large enough, $\bar{N}_{1}(t)$ will exceed $K$ for some $t\in (t^{*},t^{*}+\tau)$, deriving a contradiction. Hence, $\bar{N}_{0}(t)$ is bounded. 

    For model $(\ref{second setup})$, we have
    \begin{align*}
        &\bar{N}_{0}' = (2a(\bar{N}_{1})-1)p\bar{N}_{0} - \mu \bar{N}_{0} \\
        &\bar{N}_{1}' = 2(1-a(\bar{N}_{1}))p\bar{N}_{0} - \nu \bar{N}_{1}.
    \end{align*}
    
    Analogously, we deduce 
    \begin{align*}
        \bar{N}_{0}(t) = \bar{n}_{0}\exp\bigg(\int_{0}^{t}(2a(\bar{N}_{1}(s))-1)p-\mu ds\bigg).
    \end{align*}

    Hence,
    \begin{align*}
        \bar{N}_{1}'(t) &= 2(1-a(\bar{N}_{1}(t)))p\bar{n}_{0}\exp\bigg(\int_{0}^{t}(2a(\bar{N}_{1}(t))-1)p-\mu ds\bigg) - \nu \bar{N}_{1}(t) \\
        &\leq 2p\bar{n}_{0}\exp\bigg(\int_{0}^{t}2a(\bar{N}_{1}(t))p-p-\mu ds\bigg).
    \end{align*}

    Since $a(y) \downarrow 0$ as $y\to\infty$, we argue analogously that $\bar{N}_{1}(t)$ is bounded by some constant $K$. The growth rate of $\bar{N}_{0}$ is lower bounded by $(2a(K)-1)p-\mu$ and we argue analogously that for any $M,\tau > 0$ there exists $t^{*}$ such that for all $t\in (t^{*},t^{*}+\tau)$,
    \begin{align*}
        \bar{N}_{1}'(t) \geq 2(1-a(0))pM - \nu \bar{N}_{1}(t).
    \end{align*}
    Hence, $\bar{N}_{0}(t)$ is bounded.
\end{proof}
\begin{theorem}\label{thm: global stability}
    $\Bar{\mathbf{N}}(t) \to \Bar{\mathbf{N}}^{*}$ as $t\to\infty$ with any fixed initial condition $\Bar{\mathbf{n}}$. Moreover, for any fixed $\eta \in (\lambda,0)$ and initial data $\bar{\mathbf{n}}>0$, $\|\Bar{\mathbf{N}}(t) - \Bar{\mathbf{N}}^{*}\|=o(\exp(-|\eta|t))$ as $t\to\infty$.
\end{theorem}

\begin{proof}
    By Proposition $\ref{prop: bounded trajectory}$, the trajectory of FLLN lies in a compact set $\mathcal{K}$. To rule out periodic orbits, we denote the trapping region by $\mathcal{U} = \{(x,y)\mid x > 0,y>0\}$. Then by Dulac's theorem (Theorem $7.2.5$ in \cite{schaeffer2018ordinary}), we have for model $(\ref{first setup})$,
\begin{align*}
    \nabla\cdot \Big[\frac{-1}{x}\mathbf{F}(x,y)\Big] = -2(1-a)p'(y)+\frac{\nu}{x} > 0 \quad \text{ on } \,\,\mathcal{U}.
\end{align*}

For model $(\ref{second setup})$, since $-a'(y)y < (1-a(y))$, we have
\begin{align*}
    \nabla\cdot \Big[\frac{-1}{xy}\mathbf{F}(x,y)\Big] = \frac{2a'(y)py+2(1-a(y))p}{y^{2}} > 0 \quad \text{ on }\,\, \mathcal{U}.
\end{align*}

Therefore, there are no periodic orbits in $\mathcal{U}$. In addition, since the forward orbit through $\Bar{\mathbf{n}}$ lies in a compact subset of $\mathcal{U}$ and $\Bar{\mathbf{N}}^{*}$ is locally stable (Proposition $\ref{prop: local stability}$), we invoke Strong Poincar\'{e}–Bendixson Theorem (Theorem $7.2.5$ in \cite{schaeffer2018ordinary}) to conclude $\Bar{\mathbf{N}}^{*}$ is globally stable. Consequently, for any fixed initial data $\Bar{\mathbf{n}}>0$ and $\eta \in (\lambda,0)$, we have
\begin{align*}
    \|\Bar{\mathbf{N}}(t)-\Bar{\mathbf{N}}^{*}\| = O(\exp(-|\eta|t)) \text{ as } t\to\infty.
\end{align*}
Since $(\lambda,0)$ is open, we can strengthen the statement. For any fixed $\Bar{\mathbf{n}}$ and $\eta \in (\lambda,0)$,
\begin{align*}
    \|\Bar{\mathbf{N}}(t)-\Bar{\mathbf{N}}^{*}\| = o(\exp(-|\eta|t)) \text{ as } t\to\infty.
\end{align*}
\end{proof}

\subsection{Large-time Asymptotics of the FCLT Dynamics}
Recall $\hat{\mathbf{n}}\sim N(\mathbf{u},U)$ and $A(t) = \nabla\mathbf{F}(\bar{\mathbf{N}}(t))$. The mean function $\mathbf{m}(t) = \mathbb{E}(\hat{\mathbf{N}}(t))$ is
\begin{align*}
    \mathbf{m}(t) = \Phi(t)\mathbf{u}; \quad \Phi'(t) = A(t)\Phi(t), \quad \Phi(0) = I.
\end{align*}
The autocovariance function $\rho=\mathbb{E}(\hat{\mathbf{N}}(s)\hat{\mathbf{N}}^{\top}(t))$ has the following form:
\begin{align*}
    \rho(s,t) = \Phi(s)[U + \int_{0}^{s\wedge t}\Phi^{-1}(u)\sigma(u)\sigma^{\top}(u)[\Phi^{-1}]^{\top}(u)du]\Phi^{\top}(t); \quad s,t\geq 0.
\end{align*}
For results in this section, we assume that  $A^{*} = \nabla\mathbf{F}(\bar{\mathbf{N}}^{*})$ has distinct eigenvalues, which is equivalent to 
\begin{align*}
    \begin{cases}
        -2(1-a)(2a-1)p(\bar{N}_{1}^{*})p'(\bar{N}_{1}^{*})\bar{N}_{0}^{*} \neq [(1-a)p'(\bar{N}_{1}^{*})\bar{N}_{0}^{*}-\frac{\nu}{2}]^{2} \text{, for model } (\ref{first setup}); \\
        -4a'(\bar{\mathbf{N}}_{1}^{*})(1-a(\bar{N}_{1}^{*}))p^{2}\bar{N}_{0}^{*} \neq [a'(\bar{N}_{1}^{*})p\bar{N}_{0}^{*}+\frac{\nu}{2}]^{2} \text{, for model } (\ref{second setup}).
    \end{cases}
\end{align*}

To study large-time behaviors of the FCLT dynamics $\Hat{\mathbf{N}}$, we fix $M > 0$ and focus on the dynamics in time interval $[T,T+M]$. Recall we define time-shifted FCLT dynamics $\mathbf{G}_{T}(t)=\Hat{\mathbf{N}}(T+t)$ and the Gaussian process $\mathbf{G}^{*}$ with mean function $\mathbf{m}^{*}(t) = \mathbf{0}$ and autocovariance function $K^{*}(s,t)=V^{*}\exp((t-s)[A^{*}]^{\top})$ for $s\leq t$.
As we increase $T$ to infinity, the FCLT dynamics $\Hat{\mathbf{N}}$ in the moving interval $[T,T+M]$ will behave like the Gaussian process $\mathbf{G}^{*}$. Also recall that we cast the processes to the space of Gaussian measures on $\mathcal{H}=L^{2}([0,M],\mathcal{B}([0,M]),Leb)$, as $\mathcal{G}_{T}$'s and $\mathcal{G}_{*}$. We consider the $2-$Wasserstein distance of these two Gaussian measures in the space $\mathcal{H}$. 

\begin{lemma}\label{lemma: variance convergence}
    For any fixed $\bar{\mathbf{n}}> 0$, $\|\mathbf{m}(t)\|_{2}=O(\exp(-|\lambda|t))$ as $t\to\infty$ and $\|V(t)-V^{*}\|_{F} = o(\exp(-|\eta|t))$ as $t\to\infty$ for all $\eta \in (\lambda,0)$.
\end{lemma}
\begin{proof}
    Since $\mathbf{F}$ is $\mathcal{C}^{2}$, $\nabla \mathbf{F}$ is locally Lipschitz. By Proposition $\ref{prop: bounded trajectory}$, the trajectory of $\Bar{\mathbf{N}}$ lies in a compact set. Hence, there exists a constant $C$ such that for all $\eta \in (\lambda,0)$,
\begin{equation*}
\begin{aligned}
    &\|A(t)-A^{*}\|_{F}=\|\nabla \mathbf{F}(\Bar{\mathbf{N}}(t))-\nabla \mathbf{F}(\Bar{\mathbf{N}}^{*})\|_{F}\leq C\|\Bar{\mathbf{N}}(t)-\Bar{\mathbf{N}}^{*}\|_{2} = o(\exp(-|\eta|t)) \text{ as } t\to\infty.
\end{aligned}
\end{equation*}

Define $R(t)=A(t)-A^{*}$, we have $\int_{0}^{\infty}\|R(t)\|_{F}dt  < \infty$. Therefore, the fundamental matrix for the mean function satisfies 
\begin{align*}
    \Phi'(t) = [A^{*}+R(t)]\Phi(t); \quad \Phi(0) = I_{2} \text{, the } 2\times 2 \text{ identity matrix}.
\end{align*}

Since $A^{*}$ has distinct eigenvalues, it is diagonalizable. We apply Theorem $2.7$ (Levinson's Fundamental Theorem) in \cite{bodine2015asymptotic} to conclude as $t\to\infty$,
\begin{align*}
    \Phi(t) = [1+o(1)]\exp(A^{*}t).
\end{align*}
Hence, $\|\mathbf{m}(t)\|_{2} = \|[1+o(1)]\exp(A^{*}t)\mathbf{u}\|_{2} = O(\exp(-|\lambda|t))$ as $t\to\infty$.

Define $\Sigma(t) = \sigma(t)\sigma^{\top}(t)$ and
\begin{align*}
    W(t) = \begin{pmatrix}
        V_{1,1}(t) \\V_{1,2}(t) \\V_{2,2}(t)
    \end{pmatrix};
    B(t) = \begin{pmatrix}
        2A_{1,1}(t) &2A_{1,2}(t) &0 \\
        A_{2,1}(t) &A_{1,1}(t)+A_{2,2}(t) &A_{1,2}(t) \\
        0 &2A_{2,1}(t) &2A_{2,2}(t)
    \end{pmatrix};
    S(t) = \begin{pmatrix}
        \Sigma_{1,1}(t) \\
        \Sigma_{1,2}(t) \\
        \Sigma_{2,2}(t)
    \end{pmatrix}.
\end{align*}

The variance function satisfies $W'(t) = B(t)W(t) + S(t)$. By exponential stability of $\bar{\mathbf{N}}$ (Theorem $\ref{thm: global stability}$), we have for all $\eta \in (\lambda,0)$,
\begin{align*}
    \|B(t)-B^{*}\|_{F} = o(\exp(-|\eta|t)) \text{ and } \|S(t)-S^{*}\|_{F} = o(\exp(-|\eta|t))\text{ as } t\to\infty.
\end{align*}

By continuity of $p,a$, and $\bar{\mathbf{N}}$, $\max\{\Re(\sigma(B(t)))\} \to \max\{\Re(\sigma(B^{*}))\}=2\lambda$ as $t\to\infty$. Therefore, there exists a large enough $\tau$ such that for all $t\geq \tau$, $\max\{\Re(\sigma(B(t)))\}< \lambda < 0$. Moreover, since $S(t)\to S^{*}$ as $t\to\infty$, $\|S(t)\|_{F}$ is bounded. Hence, $\|W(t)- W^{*}\|_{2}$ is bounded.

To analyze the rate of convergence, we write the system in terms of the difference
\begin{align*}
    (W(t)-W^{*})' = (B^{*}+B(t)-B^{*})(W(t)-W^{*}) + B(t)W^{*} + S(t).
\end{align*}

Define $C(t) = (B(t)-B^{*})(W(t)-W^{*})+B(t)W^{*} + S(t)$. Then,
\begin{equation}\label{eq: variance equation}
    (W(t)-W^{*})' = B^{*}(W(t)-W^{*}) + C(t).
\end{equation}

By exponential stability of $\bar{\mathbf{N}}$ (Theorem $\ref{thm: global stability}$) and the fact $B^{*}W^{*} + S^{*} = \mathbf{0}$, we obtain that for all $\eta \in (\lambda,0)$,
\begin{align*}
    \|B(t)W^{*} + S(t)\|_{2}=o(\exp(-|\eta|t))\text{ as } t\to\infty.
\end{align*}

Since $\|W(t)-W^{*}\|_{2}$ is bounded, for all $\eta \in (\lambda,0)$,
\begin{align*}
    \|(B(t)-B^{*})(W(t)-W^{*})\|_{2} = o(\exp(-|\eta|t)) \text{ as } t\to\infty.
\end{align*}

Therefore, $\|C(t)\|_{2} = o(\exp(-|\eta|t))$ as $t\to\infty$ for all $\eta \in (\lambda,0)$.

The solution to Equation $(\ref{eq: variance equation})$ is
\begin{align*}
    &W(t)-W^{*} = \exp(B^{*}t)(W(0)-W^{*})+\int_{0}^{t}\exp(B^{*}(t-s))C(s)ds \\
    \Longrightarrow &\|W(t)-W^{*}\|_{2} \leq \|\exp(B^{*}t)(W(0)-W^{*})\|_{2}+\int_{0}^{t}\|\exp(B^{*}(t-s))C(s)\|_{2}ds.
\end{align*}
The first term is $O\exp(-2|\lambda|t)$ as $t\to\infty$. Recall $\|\exp(B^{*}t)\|_{F} = O(\exp(-2|\lambda|t))$ and $\|C(t)\|_{2} = o(\exp(-|\eta|t))$ as $t\to\infty$ for all $\eta \in (\lambda,0)$. Since $\|\exp(B^{*}t)\|_{F}$ and $\|C(t)\|_{2}$ are continuous in $t$, for a fixed $\eta \in (\lambda,0)$, there exist $M_{1},M_{2} > 0$ such that for all $t \geq 0$,
\begin{align*}
\|\exp(B^{*}t)\|_{F} \leq M_{1}\exp(-2|\lambda|t); \|C(t)\|_{2} \leq M_{2}\exp(-|\eta|t).
\end{align*}
Hence,
\begin{align*}
\int_{0}^{t}\|\exp(B^{*}(t-s))C(s)\|_{2}ds &\leq \int_{0}^{t}\|\exp(B^{*}(t-s))\|_{F}\|C(s)\|_{2}ds \\
    &\leq M_{1}M_{2}\exp(-2|\lambda|t)\int_{0}^{t}\exp((-|\eta|+2|\lambda|)s)ds \\
    &= M_{1}M_{2}\frac{\exp((-|\eta|t)-\exp(-2|\lambda|t)}{2|\lambda|-|\eta|} \\
    &= O(\exp(-|\eta|t)) \text{ as } t\to\infty.
\end{align*}
Since this statement holds for all $\eta \in (\lambda,0)$, we have as $t\to\infty$,
\begin{align*}
    \int_{0}^{t}\|\exp(B^{*}(t-s))C(s)\|_{2}ds=o(\exp(-|\eta|t)), \quad \forall \eta \in (\lambda,0).
\end{align*}

Finally, we conclude as $t\to\infty$
\begin{align*}
    &\|W(t)-W^{*}\|_{2} = o(\exp(-|\eta|t)), \quad  \forall \eta \in (\lambda,0) \\
    \Longrightarrow &\|V(t)-V^{*}\|_{F} = o(\exp(-|\eta|t)),\quad  \forall \eta \in (\lambda,0).
\end{align*}
\end{proof}

\begin{theorem}\label{theorem4}
   $W(\mathcal{G}_{T},\mathcal{G}_{*}) = o(\exp(\frac{\eta}{2}T))$ as $T\to\infty$, for all $\eta \in (\lambda,0)$.
\end{theorem}

\begin{proof}
Denote the autocovariance function of $\hat{\mathbf{N}}$ by $\rho(\cdot,\cdot)$. Then for $s\leq t$, we have 
\begin{align*}
    \rho(s,t) = \Phi(s)[V(0)+\int_{0}^{s}\Phi^{-1}(u)\Sigma(u)[\Phi^{-1}]^{\top}(u)du]\Phi^{\top}(t).
\end{align*}

When $s = t$, we have the following representation for the variance function:
\begin{align*}
    V(t) = \Phi(t)[V(0)+\int_{0}^{t}\Phi^{-1}(u)\Sigma(u)[\Phi^{-1}(u)]^{\top}du]\Phi^{\top}(t).
\end{align*}

Therefore, the autocovariance function for $G_{T}$ is
\begin{align*}
    K_{T}(s,t) &= \Phi(s+T)[V(0)+\int_{0}^{s+T}\Phi^{-1}(u)\Sigma(u)[\Phi^{-1}(u)]^{\top}du]\Phi^{\top}(t+T) \\
    &= \Phi(s+T)[\Phi^{-1}(T)V(T)[\Phi^{-1}(T)]^{\top}+\int_{T}^{s+T}\Phi^{-1}(u)\Sigma(u)[\Phi^{-1}(u)]^{\top}du]\Phi^{\top}(t+T) \\
    &= \Phi(s+T)[\Phi^{-1}(T)V(T)[\Phi^{-1}(T)]^{\top}+\int_{0}^{s}\Phi^{-1}(u+T)\Sigma(u+T)[\Phi^{-1}(u+T)]^{\top}du]\Phi^{\top}(t+T).
\end{align*}

Since $\Phi(t) = [I+o(1)]\exp(A^{*}t)$ and $\|V(t)-V^{*}\|_{F}=o\exp(-|\eta|t)$ as $t\to\infty$ for all $\eta \in (\lambda,0)$ (Lemma $\ref{lemma: variance convergence}$), as $T\to\infty$,
\begin{align*}
    \sup_{s,t\in [0,M]}\|\Phi(s+T)\Phi^{-1}(T)V(T)[\Phi^{-1}(T)]^{\top}\Phi^{\top}(t+T)-\exp(A^{*}s)V^{*}\exp([A^{*}]^{\top}t)\|_{F} = o(\exp(-|\eta|T)).
\end{align*}

In Lemma $\ref{lemma: variance convergence}$, we showed $\|\Sigma(t)-\Sigma^{*}\|_{F}=o\exp(-|\eta|t)$ as $t\to\infty$, for all $\eta \in (\lambda,0)$. By the same logic, as $T\to\infty$, the second term converges uniformly for all $s,t\in [0,M]$ with exponential rate $|\eta|$ to 
\begin{align*}
    &\exp(A^{*}s)[\int_{0}^{s}\exp(-A^{*}u)\Sigma^{*}\exp(-[A^{*}]^{\top}u)du]\exp([A^{*}]^{\top}t) \\
    = &V^{*}\exp([A^{*}]^{\top}(t-s))-\exp(A^{*}s)V^{*}\exp([A^{*}]^{\top}t).
\end{align*}

The last equality is due to the fact $\Sigma^{*} = -A^{*}V^{*}-V^{*}[A^{*}]^{\top}$. We have shown that for all $\eta \in (\lambda,0)$,
\begin{equation} \label{eq: uniform conv. kernel}
    \sup_{s,t\in [0,M]}\|K_{T}-K_{*}\|_{F} = o(\exp(-|\eta|T)) \text{ as } T\to\infty.
\end{equation}

\begin{comment}
\begin{align*}
    K_{T}(s,t) \to \exp(A^{*}s)V^{*}\exp([A^{*}]^{\top}t)+\exp(A^{*}s)[\int_{0}^{s}\exp(-A^{*}u)\Sigma^{*}\exp(-[A^{*}]^{\top}u)]du]\exp([A^{*}]^{\top}t).
\end{align*}

Using the relation $\Sigma^{*}=-A^{*}V^{*}-V^{*}[A^{*}]^{\top}$, the integral term simplifies to
\begin{align*}
    -\int_{0}^{s}\exp(-A^{*}u)[A^{*}V^{*}+V^{*}[A^{*}]^{\top}]\exp(-[A^{*}]^{\top}u)]du = \int_{0}^{s}\frac{d}{du}\exp(-A^{*}u)V^{*}\exp(-[A^{*}]^{\top}u)du.
\end{align*}

We have shown for all $s\leq t$,
\begin{align*}
    K_{T}(s,t) \to V^{*}\exp([A^{*}]^{\top}(t-s)) = K_{*}(s,t) \text{ as } T\to\infty.
\end{align*}
\end{comment}

Let $M>0$ be fixed and define $\mathcal{H}=L^{2}\left(([0,M],\mathcal{B}([0,M]),Leb); (\mathbb{R}^{2},\|\cdot\|_{2})\right)$. Recall in Section $\ref{Wasserstein}$, we have induced Gaussian measures $\mathcal{G}_{T}$ and $\mathcal{G}_{*}$ on $\mathcal{H}$. Let $\|\cdot\|_{Tr}$ denote the trace norm and $\|\cdot\|_{HS}$ denote the Hilbert-Schmidt norm. Then the $2-$Wasserstein distance between $\mathcal{G}_{T}$ and $\mathcal{G}_{*}$ is  
\begin{align*}
    W_{2}^{2}(\mathcal{G}_{T},\mathcal{G}_{*}) &= \|\mathbf{m}_{T}\|_{\mathcal{H}}^{2}+Tr(C_{T}+C_{*}-2(C_{*}^{1/2}C_{T}C_{*}^{1/2})^{1/2}) \\
    &\leq \|\mathbf{m}_{T}\|_{\mathcal{H}}^{2}+\|C_{T}+C_{*}-2(C_{*}^{1/2}C_{T}C_{*}^{1/2})^{1/2}\|_{Tr}.
\end{align*}
Since all covariance operators are positive, by operator monotonicity of the square root function (page $2$ of \cite{phillips1987uniform}), 
\begin{align*}
    \|C_{T}+C_{*}-2(C_{*}^{1/2}C_{T}C_{*}^{1/2})^{1/2}\|_{Tr} \leq \|(C_{T}+C_{*})^{2}-4(C_{*}^{1/2}C_{T}C_{*}^{1/2})\|_{Tr}^{\frac{1}{2}}.
\end{align*}

Since $C_{*}$ is positive, self-adjoint, and of trace class, $C_{*}^{\frac{1}{2}}$ is positive, self-adjoint, and Hilbert-Schmidt. This implies that $(C_{T}+C_{*})^{2}-4(C_{*}^{1/2}C_{T}C_{*}^{1/2})$ is a positive operator. To see this, let $(e_{i})$ be an orthonormal eigenbasis for $C_{*}^{\frac{1}{2}}$ corresponding to eigenvalues $(\lambda_{i})$. Then, we have
\begin{align*}
    \langle C_{*}^{1/2}C_{T}C_{*}^{1/2}e_{i},e_{i}\rangle &= \langle C_{T}C_{*}^{1/2}e_{i},C_{*}^{1/2}e_{i}\rangle \\
    &= \lambda_{i}^{2}\langle C_{T}e_{i},e_{i}\rangle \\
    &= \langle C_{*}C_{T}e_{i},e_{i}\rangle = \langle C_{T}C_{*}e_{i},e_{i}\rangle.
\end{align*}
Therefore, 
\begin{align*}
    \langle(C_{T}+C_{*})^{2}-4(C_{*}^{1/2}C_{T}C_{*}^{1/2})e_{i},e_{i}\rangle = \langle(C_{T}-C_{*})^{2})e_{i},e_{i}\rangle \geq 0.
\end{align*}
Since the trace norm of a positive operator is equal to to its trace, we have
\begin{align*}
    \|(C_{T}+C_{*})^{2}-4(C_{*}^{1/2}C_{T}C_{*}^{1/2})\|_{Tr}^{\frac{1}{2}} &= Tr((C_{T}+C_{*})^{2}-4(C_{*}^{1/2}C_{T}C_{*}^{1/2}))^{\frac{1}{2}}\\
    &= Tr((C_{T}-C_{*})^{2})^{\frac{1}{2}} \\
    &= \|C_{T}-C_{*}\|_{HS}.
\end{align*}

Denote the space of Hilbert-Schmidt operators on $\mathcal{H}$ by $B_{HS}(\mathcal{H})$, we have the following isometric isomorphisms
\begin{align*}
    B_{HS}(\mathcal{H}) \cong \mathcal{H}^{*}\otimes \mathcal{H} \cong \mathcal{H}\otimes \mathcal{H} \cong L^{2}([0,M]^{2}; (\mathbb{R}^{2},\|\cdot\|_{2})\otimes(\mathbb{R}^{2},\|\cdot\|_{2})),
\end{align*}
where $\mathcal{H}^{*}$ is the dual of $\mathcal{H}$.

Therefore,
\begin{align*}
    \|C_{T}-C_{*}\|_{HS}^{2} &= \|K_{T}-K_{*}\|_{L^{2}([0,M]^{2}; (\mathbb{R}^{2},\|\cdot\|_{2})\otimes(\mathbb{R}^{2},\|\cdot\|_{2}))}^{2} \\
    &= \int_{0}^{M}\int_{0}^{M}\|K_{T}(s,t)-K_{*}(s,t)\|^{2}_{(\mathbb{R}^{2},\|\cdot\|_{2})\otimes(\mathbb{R}^{2},\|\cdot\|_{2})}dsdt\\
    &\leq \text{constant}\cdot \int_{0}^{M}\int_{0}^{M}\|K_{T}(s,t)-K_{*}(s,t)\|^{2}_{F}dsdt \\
    &= 2\cdot\text{constant}\cdot \int_{0}^{M}\int_{s}^{M}\|K_{T}(s,t)-K_{*}(s,t)\|^{2}_{F}dtds.
\end{align*}

By Equation $\ref{eq: uniform conv. kernel}$, we deduce as $T\to\infty$,
\begin{align*}
    \|C_{T}-C_{*}\|_{HS}= o(\exp(-|\eta|T)), \quad \forall \eta \in (\lambda,0).
\end{align*}

From Lemma $\ref{lemma: variance convergence}$, we deduce $\|\mathbf{m}_{T}\|_{\mathcal{H}}^{2} = O(\exp(-2|\lambda|T))$ as $T\to\infty$. Hence,
\begin{align*}
    W_{2}^{2}(\mathcal{G}_{T},\mathcal{G}_{*}) &\leq \|\mathbf{m}_{T}\|_{\mathcal{H}}^{2}+\|C_{T}-C_{*}\|_{HS}\\
    &= o(\exp(-|\eta|T)) \text{ as } T\to\infty, \quad \forall \eta \in (\lambda,0).
\end{align*}

We conclude, $W_{2}(\mathcal{G}_{T},\mathcal{G}_{*}) = o(\exp(\frac{\eta}{2}T))$ as $T\to\infty$ for all $\eta \in (\lambda,0)$.
\end{proof}

\bibliography{mybib}

\end{document}